\documentclass[journal]{IEEEtran}
\ifCLASSINFOpdf
\else
\fi

\usepackage{amsmath,amssymb,amsfonts}
\usepackage{mathtools}
\usepackage{bm}
\usepackage{amsthm}
\usepackage{nicefrac}
\newtheorem{proposition}{Proposition}
\newtheorem{lemma}{Lemma}
\newtheorem{corollary}{Corollary}

\usepackage{graphicx}
\usepackage{subcaption}
\graphicspath{{images/}}
\usepackage{color}
\definecolor{mygreen}{RGB}{28,172,0} 
\definecolor{mylilas}{RGB}{170,55,241}
\usepackage{wrapfig}

\usepackage{tikz}
\usetikzlibrary{math}
\usepackage{pgfplots}
\pgfplotsset{compat=1.16}

\usepackage{tabularx}
\usepackage{multirow}
\usepackage{longtable}

\usepackage{hyperref}
\usepackage{cleveref}

\usepackage{cite}
\usepackage{url}
\usepackage{dirtytalk}

\usepackage{textcomp}
\usepackage{gensymb}
\usepackage{parskip}
\usepackage{comment}
\usepackage{stmaryrd }
\usepackage{pdfpages}
\usepackage{colortbl}
\usepackage{titlesec}
\usepackage{caption}
\usepackage{enumerate} 

\usepackage[noend]{algpseudocode}
\usepackage{algorithm}

\newcommand{\pa}[1]{\left(#1\right)} 
\newcommand{\br}[1]{\left[#1\right]} 
\newcommand{\acc}[1]{\left\{#1\right\}} 
\newcommand{\abs}[1]{\left|#1\right|} 

\newcommand{\st}{\text{s.t.}}

\newcommand{\Tr}[1]{\text{Tr}\acc{#1}} 
\newcommand{\norm}[2]{\left\|#1\right\|_{#2}} 
\newcommand{\normw}[2]{\left\|#1\right\|_{#2,\mathbf{w}}} 
\newcommand{\sr}[1]{\rho\pa{#1}} 
\newcommand{\sigmsq}[1]{\sigma_{\max}^2\pa{#1}} 

\newcommand{\argmax}[1]{\underset{#1}{\arg \max}}
\newcommand{\argmin}[1]{\underset{#1}{\arg \min}}

\newcommand{\Rate}[1]{R_{#1}\pa{\mathbf{Q}_{#1},\mathbf{Q}_{-#1}}} 
\newcommand{\EE}[1]{\text{EE}_{#1}\pa{\mathbf{Q}_{#1},\mathbf{Q}_{-#1}}} 
\newcommand{\Ge}{\mathcal{G}_E} 
\newcommand{\Ges}{\mathcal{G}_E^\star} 
\newcommand{\C}[2]{\mathbb{C}^{#1\times#2}} 
\newcommand{\R}[2]{\mathbb{R}^{#1\times#2}} 

\newcommand{\Q}[1]{\mathcal{Q}_q\pa{#1}} 
\newcommand{\Qeq}[1]{\partial\mathcal{Q}_q\pa{#1}} 
\newcommand{\Qbar}[1]{\overline{\mathcal{Q}}_q\pa{#1}} 
\newcommand{\Qbareq}[1]{\partial\overline{\mathcal{Q}}_q\pa{#1}} 

\newcommand{\Ri}[1]{\mathbf{R}_{-#1}\pa{\mathbf{Q}_{-#1}}}  
\newcommand{\Rm}[1]{\mathbf{R}_{-#1}^{-1}\pa{\mathbf{Q}_{-#1}}} 

\newcommand{\Hb}[1]{\overline{\mathbf{H}}_{#1}} 
\newcommand{\Qb}[1]{\overline{\mathbf{Q}}_{#1}} 
\newcommand{\Rib}[1]{\overline{\mathbf{R}}_{-#1}\pa{\overline{\mathbf{Q}}_{-#1}}} 
\newcommand{\Rmb}[1]{\overline{\mathbf{R}}_{-#1}^{-1}\pa{\overline{\mathbf{Q}}_{-#1}}} 
\newcommand{\Intb}[1]{\Hb{#1#1}^H\Rmb{#1}\Hb{#1#1}} 
\newcommand{\invIntb}[1]{\pa{\Hb{#1#1}^H\Rmb{#1}\Hb{#1#1}}^{-1}} 
\newcommand{\Xb}[1]{\overline{\mathbf{X}}_{#1}\pa{\Qb{-#1}}} 
\newcommand{\EEb}[1]{\overline{\text{EE}}_{#1}\pa{\Qb{#1},\Qb{-#1}}} 
\newcommand{\Rateb}[1]{\overline{R}_{#1}\pa{\Qb{#1},\Qb{-#1}}} 
\newcommand{\Geb}{\overline{\mathcal{G}}_E} 
\newcommand{\PEEb}[1]{\hat{P}_{#1}\pa{\Qb{-#1}}}
\newcommand{\PEEbNE}[1]{\hat{P}_{#1}\pa{\QbNE{-#1}}}
\newcommand{\PEEbp}[1]{\hat{P}_{#1}\pa{\Qb{-#1}^{'}}}
\newcommand{\Pub}[1]{P_{#1}^{*}\pa{\Qb{-#1}}} 

\newcommand{\Qaggb}{\overline{\mathbf{Q}}}
\newcommand{\QaggbNE}{\overline{\mathbf{Q}}^{\text{NE}}}
\newcommand{\QbNE}[1]{\overline{\mathbf{Q}}^{\text{NE}}_{#1}}
\newcommand{\XbNE}[1]{\overline{\mathbf{X}}_{#1}\pa{\QbNE{-#1}}} 

\newcommand{\setAggb}{\partial\overline{\mathcal{Q}}(\mathbf{p})}

\newcommand{\Fb}[1]{\overline{\mathbf{F}}_{#1}\pa{\Qaggb}} 
 
\newcommand{\Fbp}[1]{\overline{\mathbf{F}}_{#1}\pa{\Qaggb^{'}}}

\newcommand{\Faggb}{\overline{\mathbf{F}}\pa{\Qaggb}} 
\newcommand{\Faggbp}{\overline{\mathbf{F}}\pa{\Qaggb^{'}}}
 
\newcommand{\Eb}[1]{\overline{\mathbf{E}}_{#1}}
\newcommand{\Sc}{\overline{\mathbf{S}}}

\newcommand{\Gb}[2]{\overline{\mathbf{G}}_{#1}\pa{#2}}

\newcommand{\Db}{\overline{\mathbf{\Delta}}}

\newcommand{\RmbD}[1]{\overline{\mathbf{R}}_{-#1}^{-1}\pa{\Db}} 
\newcommand{\invIntbD}[1]{\pa{\Hb{#1#1}^H\RmbD{#1}\Hb{#1#1}}^{-1}} 
 
\newcommand{\Scb}{\overline{\overline{\mathbf{S}}}}

\newcommand{\PEEaggb}{\hat{\mathbf{p}}\pa{\Qaggb}} 
\newcommand{\PEEaggbp}{\hat{\mathbf{p}}\pa{\Qaggb^{'}}} 
\newcommand{\setEEAggb}{\partial\overline{\mathcal{Q}}\pa{\hat{\mathbf{p}}}}
\newcommand{\setEEAggbs}{\partial\overline{\mathcal{Q}}\pa{\Qaggb}}
\newcommand{\setEEAggbp}{\partial\overline{\mathcal{Q}}\pa{\hat{\mathbf{p}}^{'}}}
\newcommand{\setEEAggbps}{\partial\overline{\mathcal{Q}}\pa{\Qaggb^{'}}}
\newcommand{\fullsetEEAggb}{\overline{\mathcal{Q}}\pa{\mathbf{p}}}
\newcommand{\setEEAggbNE}{\partial\overline{\mathcal{Q}}\pa{\QaggbNE}}
\newcommand{\QVI}{\text{QVI}\pa{\partial\overline{\mathcal{Q}},\overline{\mathbf{F}}}}
\newcommand{\BREEb}[1]{\overline{\textbf{\text{BR}}}_{#1}\pa{\Qb{-#1}}} 
\newcommand{\BREEAggb}{\overline{\textbf{\text{BR}}}\pa{\Qaggb}} 
\newcommand{\BREEAggbp}{\overline{\textbf{\text{BR}}}\pa{\Qaggb^{'}}} 
\newcommand{\Xaggb}{\overline{\mathbf{X}}\pa{\Qaggb}} 
\newcommand{\Xaggbp}{\overline{\mathbf{X}}\pa{\Qaggb^{'}}}
\newcommand{\FaggbNEH}{\overline{\mathbf{F}}^H\pa{\QaggbNE}}

\newif
\ifarchiv
\archivtrue

\begin{document}
%
\title{Energy Efficient Competitive Resource Allocation in MIMO networks}
%
%
%

\author{Guillaume~Thiran,~\IEEEmembership{Member,~IEEE,}
        Ivan~Stupia,~\IEEEmembership{Member,~IEEE,}
        and~Luc~Vandendorpe,~\IEEEmembership{Fellow,~IEEE}
\thanks{G. Thiran, I. Stupia and L. Vandendorpe are with the Université Catholique de Louvain,
B-1348 Louvain-la-Neuve, Belgium (e-mail: guillaume.thiran@uclouvain.be;  ivan.stupia@uclouvain.be;
luc.vandendorpe@uclouvain.be). }
\thanks{GT is a Research Fellow of the Fonds de la Recherche Scientifique - FNRS.}

}

\maketitle

\begin{abstract}
This paper considers the competitive resource allocation problem in Multiple-Input Multiple-Output (MIMO) interfering channels, when users maximize their energy efficiency. Considering each transmitter-receiver pair as a selfish player, conditions on the existence and uniqueness of the Nash equilibrium of the underlying noncooperative game are obtained. A decentralized asynchronous algorithm is proven to converge towards this equilibrium under the same conditions. Two frameworks are considered for the analysis of this game. On the one hand, the game is rephrased as a Quasi-Variational Inequality (QVI). On the other hand, the best response of the players is analyzed in light of the contraction mappings. For this problem, the contraction approach is shown to lead to tighter results than the QVI one. When specializing the obtained results to OFDM networks, the obtained conditions appear to significantly outperform  state-of-the-art works, and to lead to much simpler decentralized algorithms. Numerical results finally assess the obtained conditions in different settings.
\end{abstract}


%
\IEEEpeerreviewmaketitle

\section{Introduction}
%
%
%
%
\IEEEPARstart{T}{he} interference channel model enables to describe communication systems in which radio resources are shared among several users. The design of resource allocation schemes, either coordinated or decentralized, plays a central role in this model and lots of works have considered noncooperative policies based on game theory in the last 20 years. However, for scenarios in which fully interfering users wish to communicate simultaneously, performances of such decentralized approaches are often not good enough to balance the intrinsic defaults of competitive schemes, i.e. their non-Pareto optimality and interference-dependent convergence rate. 

Yet, solutions envisioned for future wireless networks open again the way to such approaches. Indeed, in Mesh-Edge Computing (MEC) networks, the radio resources not only serve the communication, but also support the computation offloading to nearby servers. This offloading is by definition delay-constrained and sporadic, and centralized solutions may fail to capture such asynchronous events. This is why many recent works consider joint computation and communication resource allocation through game theory, with as a channel model the interference channel \cite{li2019computation,liu2019joint,moura2018game}. Such a joint allocation obviously heavily depends on the fundamental performance limits of both the communication and computation part. Hence, this calls for an in-depth understanding of the competitive radio-resource allocation in the interference channel model.

\subsection{Related Works}
This interference channel model has been widely studied for users maximizing their rate \cite{scutari_optimal_2007,scutari_competitive_2008,ren2010distributed,etkin2007spectrum,scutari_optimal_2008, yu2002distributed,cendrillon2007autonomous,shum2007convergence,luo2006analysis,yamashita2004nonlinear,scutari2010convex,pang_design_2010,scutari_mimo_2009,scutari2009mimo}. Several frameworks have been considered, among which is the field of nonlinear complementary problems \cite{luo2006analysis,yamashita2004nonlinear}, contraction analysis \cite{scutari_competitive_2008,ren2010distributed,etkin2007spectrum,scutari_optimal_2008} and variational inequalities \cite{scutari2010convex,pang_design_2010}. Remarkably, for parallel interference channels (i.e. for OFDM/DSL games), these approaches lead to the same conditions on the feasibility of the competitive resource allocation (see \cite{scutari_competitive_2008} for an exhaustive comparison). Several flavours of decentralized algorithms (sequential, synchronous, asynchronous) known as Iterative WaterFilling Algorithms (IWFA) have been proven to converge towards the equilibrium of the game under these feasibility conditions. The MIMO interference channel model has however been far less studied. To the best of our knowledge, the contraction framework has been preferred over the other approaches \cite{scutari_mimo_2009, scutari2009mimo, scutari_competitive_2008} and has led to convergence conditions similar to those of parallel channels.\\
The extension of those results to games in which users maximize their Energy Efficiency (EE) has been considered in \cite{BacciEECompetitive, stupia_power_2015,ZapponeCompetitiveEE,miao2011distributed,zhong2013energy,buzzi2011potential}. Again, several approaches have been followed, which include the theory of standard functions \cite{ZapponeCompetitiveEE,miao2011distributed}, potential games \cite{zhong2013energy,buzzi2011potential}, quasi-variational inequalities \cite{stupia_power_2015} and generalized Nash equilibrium problems \cite{BacciEECompetitive}. Only few of the above works consider MIMO channels, as we are only aware of \cite{zhong2013energy,pan2014totally}. Yet, a unified view on the effectiveness of the different frameworks is currently lacking.

\subsection{Contributions}
The aim of this paper is to fill the above gap by analyzing competitive resource allocation in MIMO networks when users seek to maximize their EE. The study of this game through the frameworks of contraction theory and QVI, even though not sufficient to provide a unified view, already leads to interesting conclusions. Indeed,  contrarily to games in which users maximize their rate, the two frameworks lead surprisingly to different results: the contraction approach provides tighter conditions than what the QVI framework is able to. \\
The obtained conditions ensure the existence of a unique Nash equilibrium, and the convergence of an asynchronous EE-IWFA. They can be interpreted as a limitation of the Multi-User Interference (MUI), as well as a smoothness requirement on the power mapping of the players. When particularized to an OFDM network, our new conditions outperform state-of-the-art works, and a much simpler decentralized algorithm is obtained. 
\subsection{Organization}
The paper is organized as follows. In \Cref{sec:systemmodel}, the system model is described and the noncooperative game is formulated.  In \Cref{sec:BR}, the Best Response (BR) of the players is analyzed, leading to a simpler formulation of the game. In \Cref{sec:NEproblem}, the existence and uniqueness of the Nash equilibrium are analyzed through the frameworks of QVI and contraction theory. The two approaches are then compared in light of their results. Finally, in \Cref{sec:algo}, a decentralized algorithm is designed and numerical validation is performed.
\subsection{Notations}
The following notations are used throughout the paper. Lowercase,  lowercase boldface and uppercase boldface denote respectively scalars, vectors and matrices.  The closed and open interval from $a$ to $b$ are written as $[a\,,\,b]$ and $(a\,,\,b)$. $[\mathbf{A}]_{qr}$ denotes the (q,r)th element of $\mathbf{A}$.  $\R{m}{n}$ and $\C{m}{n}$ are the sets of $m\times n$ real and complex matrices. The operators $\mathcal{E}\acc{.}$, $(.)^T$, $(.)^H$, $(.)^\sharp$, $(.)^s$, $\text{det}\pa{.}$, $\Tr{.}$, $\norm{.}{p}$, $\norm{.}{F}$, $\sigma_{\max}\pa{.}$, $\sr{.}$ correspond respectively to the expectation, transpose, Hermitian transpose, Moore-Penrose pseudoinverse, symmetric part, determinant, trace, $p$-norm, Frobenius norm, maximum singular value and spectral radius. The notation $\mathbf{A} \succ \mathbf{B}$ (resp. $\succeq$, $\preceq$, $\prec$) means that $\mathbf{A}-\mathbf{B}$ is positive definite (resp. positive semi-definite, negative semi-definite, negative definite), while $>$, $\geq$, $<$, $\leq$ denote the componentwise order relations. $\text{diag}\pa{\mathbf{A}}$ produces a vector with the diagonal elements of $\mathbf{A}$. The operator $(.)^+$ is defined as $\max(.,\mathbf{0})$ where the maximum is taken componentwise if this applies to a matrix. Given two sets $\mathcal{S}_1$ and $\mathcal{S}_2$, $\mathcal{S}\triangleq\mathcal{S}_1\bigtimes\mathcal{S}_2$ is their Cartesian product. Finally, $[\mathbf{A}]_{\mathcal{S}}$  denote the projection of $\mathbf{A}$ onto the set $\mathcal{S}$ with respect to the Frobenius norm.

\section{System model}
\label{sec:systemmodel}
As in \cite{scutari_optimal_2007, scutari_optimal_2008, scutari_mimo_2009,stupia_power_2015}, we consider $Q$ Transmitter-Receiver Pairs (TRP) such that each transmitter wishes to communicate with the corresponding receiver, the set of TRPs being denoted by $\Omega \triangleq \acc{1,\hdots,Q}$. Each TRP is equipped at the emitter and receiver side with respectively $n_{T_q}$ and $n_{R_q}$ antennas for the $q$th TRP. The channel between the $r$th transmitter and the $q$th receiver is hence a MIMO channel assumed to be frequency flat, thus described by a possibly rank-deficient matrix $\mathbf{H}_{qr} \in \C{n_{R_q}}{n_{T_r}}$ of rank $r_q$. Letting $\mathbf{x}_q \in \mathbb{C}^{n_{T_q}}$ be the symbols emitted by the $q$th transmitter, the received symbols $\mathbf{y}_q \in \mathbb{C}^{n_{R_q}}$ are modeled in a complex baseband representation as 
\begin{equation}
    \mathbf{y}_q = \mathbf{H}_{qq}\mathbf{x}_q+\sum_{r\neq q}\mathbf{H}_{qr}\mathbf{x}_r + \mathbf{n}_r,
\end{equation}
with $\mathbf{n}_r \in \mathbb{C}^{n_{R_q}}$ an Additive White Gaussian Noise (AWGN) vector with covariance matrix $\mathbf{R}_{n_q}$, assumed to be nonsingular (implying $\mathbf{R}_{n_q}\succ \mathbf{0}$). At the emitter side, the covariance matrix of the symbols $\mathbf{Q}_q = \mathcal{E}\acc{\mathbf{x}_q\mathbf{x}_q^H}$
can be tuned to optimize the system performances, under the condition that the mean emitted power is smaller than a maximum power $P_q$: 
\begin{align}
    \mathcal{E}\acc{\norm{\mathbf{x}_q}{2}^2}=  \Tr{\mathbf{Q}_q} \leq P_q.
\end{align}
The hardware circuitry power consumed by the $q$th transmitter is supposed to be a constant $\Psi_q$ \cite{stupia_power_2015}. We furthermore assume that channel matrices remain constant during the whole transmission, and that each TRP has a perfect knowledge of its direct channel matrix. Moreover, each TRP is able to measure perfectly the MUI-plus-noise covariance matrix. Considering the MUI as noise, the rate achieved by the $q$th TRP is bounded by the capacity of the channel, reading as \cite{goldsmith2003capacity, scutari_mimo_2009}
 \begin{equation}
    \Rate{q} = \log\det\pa{\mathbf{I}+\mathbf{H}_{qq}^H\Rm{q}\mathbf{H}_{qq}\mathbf{Q}_q},
    \label{eq:Rate}
\end{equation}
with the MUI-plus-noise covariance matrix
\begin{equation}
    \label{eq:MUI}
    \Ri{q} \triangleq \mathbf{R}_{n_q} + \sum_{r\neq q}\mathbf{H}_{qr}\mathbf{Q}_r\mathbf{H}_{qr}^H\succ\mathbf{0},
\end{equation}
and with $\mathbf{Q}_{-q}=\acc{\mathbf{Q}_{r}}_{r\neq q}$  denoting the other players' strategies.
The EE can then be defined as
\begin{align}
    \EE{q} = \frac{\Rate{q}}{\Psi_q + \Tr{\mathbf{Q}_q}}.
\end{align}
The interaction of the $Q$ players, each of them optimizing selfishly its own EE, is thus captured by the following EE game $\mathcal{G}_E$: \begin{eqnarray}
    \pa{\mathcal{G}_E}\, :  & \begin{array}{cc}
         \underset{\mathbf{Q}_q}{\text{maximize}}& \EE{q}  \\ 
         \text{subject to} & \mathbf{Q}_q \in \Q{P_q}
    \end{array}
    & \forall q \in \Omega,
    \label{eq:EEgame}
\end{eqnarray}
with the strategy set of each user being defined\footnote{The positive semi-definiteness ensures the Hermitian property as any matrix having only real eigenvalues is Hermitian \cite[Section 7.1]{horn_matrix_1990}.} as
\begin{equation}
        \Q{P_q}\triangleq\acc{\mathbf{Q}\in\mathbb{C}^{n_{T_q}\times n_{T_q}} : \, \mathbf{Q}\succeq\mathbf{0},\, \Tr{\mathbf{Q}}\leq P_q}.
        \label{eq:adm_strat}
\end{equation}
Similarly, the set 
 \begin{equation}
    \Qeq{P_q}\triangleq\acc{\mathbf{Q}\in\mathbb{C}^{n_{T_q}\times n_{T_q}} : \, \mathbf{Q}\succeq\mathbf{0},\, \Tr{\mathbf{Q}}= P_q},
    \label{eq:adm_strat_full}
\end{equation}
is introduced, corresponding to covariance matrices using all the available power. 

\section{Best Response Analysis}
\label{sec:BR}
In this section, for the sake of completeness, we briefly recall known results on the BR of the players in case of the EE game: we i) cast the game into a full-column rank one \cite{scutari_mimo_2009}, ii) provide its BR \cite{xu_energy_2013} and translate the game into an equivalent, easier to analyze, pseudo-game \cite{stupia_power_2015}, iii) interpret the BR as a projection \cite{scutari_mimo_2009}.

\subsection{Full-Column Rank Game}
As shown in \cite{scutari_mimo_2009}, the behaviour of such MIMO game heavily depends on the rank of the channel matrices $\mathbf{H}_{qq}$, making the analysis difficult. To overcome this problem, an equivalent game formulation enables to limit the analysis to full rank square channel matrices and full-column rank rectangular matrices (hence removing the case of full-row rank matrices). Defining the compact Singular Value Decomposition (SVD) of the direct channel matrices as $\mathbf{H}_{qq} = \mathbf{U}_{q,1}\mathbf{\Sigma}_{qq}\mathbf{V}_{q,1}^H $ with $\mathbf{\Sigma}_{qq} \in \R{r_q}{r_q}$ a positive definite diagonal matrix, and with $\mathbf{U}_{q,1} \in \C{n_{R_q}}{r_q}$ and $\mathbf{V}_{q,1} \in \C{n_{T_q}}{r_q}$ being semi-unitary,
the BR of each user is indeed always such that \cite{scutari_mimo_2009}
\begin{align}
    \mathbf{Q}_q &= \mathbf{V}_{q,1}\Qb{q}\mathbf{V}_{q,1}^H,
\end{align}
with $\Qb{q}\in\Qbar{P_q}$ and with the lower-dimensional strategy sets defined as
            \begin{align}
                \Qbar{P_q}&\triangleq\acc{\mathbf{Q}\in\C{r_{q}}{r_{q}} : \, \mathbf{Q}\succeq\mathbf{0},\, \Tr{\mathbf{Q}}\leq P_q},\\
                \Qbareq{P_q}&\triangleq\acc{\mathbf{Q}\in\C{r_{q}}{r_{q}} : \, \mathbf{Q}\succeq\mathbf{0},\, \Tr{\mathbf{Q}} =  P_q}.
            \end{align}
Building on \cite[Appendix C]{scutari_mimo_2009}, defining $\Hb{qr} \triangleq \mathbf{H}_{qr}\mathbf{V}_{r,1}$ $\forall q,r\in\Omega$, $\Ge$ can be translated into $\Ges$, defined as
\begin{eqnarray}
    \pa{\mathcal{G}_E^*}\, :  & \begin{array}{cc}
         \underset{\Qb{q}}{\text{maximize}}& \EEb{q}  \\ 
         \text{subject to} & \Qb{q} \in \Qbar{P_q}
    \end{array}
    & \forall q \in \Omega,
    \label{eq:EEgame_fullrank}
\end{eqnarray}
where $\EEb{q}$ is the EE function in which all channel and strategy matrices are replaced by their lower-dimensional counterpart\footnote{One can check that $\EEb{q}=\EE{q}$.}. In this game reformulation, the matrices $\Hb{qq}$ are either square full rank (if $\mathbf{H}_{qq}$ is rectangular full-row rank or square full rank.) or rectangular full-column rank (in any other cases).

\subsection{Equivalent Pseudo-Game}
Building on the above game reformulation, considering a player $q\in\Omega$, his BR is defined as
\begin{equation}
   \BREEb{q}=  
         \argmax{\Qb{q} \in\Qbar{P_q}} \EEb{q},
   \label{eq:EE_best_response}
\end{equation}
where $\Qb{-q}$ is considered fixed in the above optimization problem. 
To the best of our knowledge, no closed-form solution exists for the above problem. Yet, for a fixed emitted power (namely, $\Tr{\Qb{q}}$ constant),  \eqref{eq:EE_best_response} boils down to the maximization of the rate, whose solution is the well-known waterfilling solution \cite{scutari_mimo_2009}. 
Several approaches \cite{stupia_power_2015,xu_energy_2013} build on this idea to obtain an expression of the BR. From their results, it appears the BR can be obtained in two steps: first, the optimal power is obtained, without taking into account the maximum power constraint. Then, once the power is fixed, the EE maximization boils down to the rate maximization which leads to the well-known waterfilling solution. To that aim, we introduce the following compact EigenValue Decomposition (EVD):
\begin{equation}
    \Intb{q} = \mathbf{U}_q\mathbf{D}_q\mathbf{U}_q^H, \label{eq:invInt_EVD}
\end{equation}
with $\mathbf{D}_q \in \R{r_q}{r_q}$ a positive definite diagonal matrix, $\mathbf{U}_q\in \C{r_q}{r_q}$ a unitary matrix. Formally, we define
\begin{equation}
   \Pub{q} = \Tr{\argmax{\mathbf{0}\preceq \Qb{q}\in \C{r_q}{r_q}} \EEb{q}},
\label{eq:EE_power_unconstrained}
\end{equation}
as the unconstrained power, which can be obtained through the Dinkelbach method \cite{dinkelbach_nonlinear_1967,schaible_fractional_1983,Zappone_EE}, described in \Cref{algo:dinkelbach}. This method converges super-linearly to the solution and has been previously used in \cite{stupia_power_2015} for EE in OFDM systems.  
\begin{algorithm}
\caption{Dinkelbach method }
\begin{algorithmic}
\State Set $i=0$, $\mathbf{0}\neq \Qb{q}^{(0)}\in \C{r_q}{r_q}$, $\epsilon \ll 1$, and $\delta^{(0)}=2\epsilon$.
    \While{$\delta^{(i)}>\epsilon$}
         \State Set 
        $$\nu^{(i+1)}=\frac{\overline{R}_q(\Qb{q}^{(i)},\Qb{-q})}{\Tr{\Qb{q}^{(i)}}+\Psi_q}.$$
        \State Set $$\Qb{q}^{(i+1)} = \mathbf{U}_q\pa{\frac{1}{\nu^{\pa{i+1}}}\mathbf{I}_{r_q}-\mathbf{D}_q^{-1}}^+\mathbf{U}_q^H, $$
        \State \hspace{0.75cm}with $\mathbf{U}_q$ and $\mathbf{D}_q$ defined in \eqref{eq:invInt_EVD}.
        \State Set \begin{align*} \delta^{(i+1)} =\Big|\overline{R}_q(\Qb{q}^{(i+1)},\Qb{-q}) \qquad \qquad \qquad\\\qquad \qquad \qquad \qquad-\nu^{(i+1)}\pa{\Tr{\Qb{q}^{(i+1)}}+\Psi_q}\Big|.
        \end{align*}  
        \State Set $i\longleftarrow i+1$.
    \EndWhile
\State Output $\Tr{\Qb{q}^{(i)}}$.
\label{algo:dinkelbach}
\end{algorithmic}
\end{algorithm}

Once the optimum unconstrained power is obtained, the optimum constrained power is given by \cite{xu_energy_2013}
\begin{align*}
    \PEEb{q}=\min\pa{P_q,\Pub{q}},
\end{align*}
and the BR can be expressed as \cite{scutari_competitive_2008,scutari_mimo_2009}, 
\begin{align}
    \BREEb{q} &= \mathbf{U_q}\pa{\mu_q\mathbf{I}_{r_q}-\mathbf{D}_q^{-1}}^{+}\mathbf{U}_q^H, \nonumber\\& \quad  \quad \quad\text{with}\, \,\mu_q \,\, \st \quad  \Tr{\Qb{q}}=\PEEb{q},
    \label{eq:waterfilling_EE}
\end{align}
with $\mathbf{U}_q$ and $\mathbf{D}_q$ coming from \eqref{eq:invInt_EVD}. Beside providing an efficient way to compute the BR, \eqref{eq:waterfilling_EE} also paves the way to the definition of an equivalent pseudo-game reading as
\begin{align}
    \pa{\Geb}\, :  & \begin{array}{cc}
         \underset{\Qb{q}}{\text{maximize}}& \Rateb{q} \\
         \text{subject to} & \Qb{q} \in \Qbareq{\PEEb{q}}
    \end{array}
      \forall q \in \Omega, 
    \label{eq:EEgame_full}
\end{align}
The difference between the above formulation and a classical game is the fact that the strategy sets of the players also depend on the other strategies\footnote{If each player independently chooses a strategy, then it may happen the strategy does not belong to the strategy sets, hence the name \textit{pseudo-game}. }. Determining the equilibria of the above game is called a Generalized Nash Equilibrium Problem (GNEP), such problems being studied in \cite{facchinei2007generalized}. We stress out that the three game formulations (i.e. $\Ge,\, \Ges, \, \Geb$) are equivalent: therefore, in the following, we will w.l.o.g analyze $\Geb$, its properties extending to $\Ge$.

\subsection{Best Response as a Projection}
The waterfilling expression \eqref{eq:waterfilling_EE} gives a very efficient way to compute the BR of each player. However, the impact of the other strategies is hidden behind the EVD of $\Intb{q}$, which makes the analysis of the competitive interactions difficult. Fortunately, the BR of the players can also be expressed as the projection of a matrix onto the strategy set, as shown in \cite[Lemma 1]{scutari_mimo_2009}: defining $\Xb{q} \triangleq -\invIntb{q}$, the BR can be expressed as\footnote{Note that with regards to the results of \cite{scutari_mimo_2009}, the constant $c_q$ is not needed as the lower-dimensional direct channel matrices are full-column rank.} 
\begin{equation}
    \BREEb{q} = \br{\Xb{q}}_{\Qbareq{\PEEb{q}}}.
    \label{eq:EE_projection}
\end{equation}



\section{Nash Equilibrium Problem}
\label{sec:NEproblem}
The above section enables all players to compute their BR, given the other strategies. In this section, we analyze the Nash Equilibria (NE) of the game, i.e. the fixed-points of the BR mapping\footnote{We only consider the pure NE of the game without loss of generality. Indeed, as the payoff functions are strictly quasi-concave, mixed strategy profiles are never a BR \cite[Theorem 33]{etkin2007spectrum}.}. Formally, defining the aggregate strategy profile, the aggregate maximum power vector and the corresponding strategy set, as well as the optimal power vector and associated strategy set\footnote{For the sake of notations, the dependency of $\hat{\mathbf{p}}$ on $\Qaggb$ is omitted in $\setEEAggb$.} as 
    \begin{eqnarray}
        \Qaggb & \triangleq&  \begin{bmatrix} \Qb{1} &  \hdots & \Qb{Q}
        \end{bmatrix}^H,\\
         \mathbf{p} & \triangleq &\begin{bmatrix} P_1 & \hdots&  P_Q \end{bmatrix}^T, \\
         \fullsetEEAggb  & \triangleq & \bigtimes_{q=1}^Q \Qbar{P_q} ,\\
        \PEEaggb & \triangleq & \begin{bmatrix} \PEEb{1} & \hdots &\PEEb{Q} \end{bmatrix}^T,\\
         \setEEAggb  & \triangleq &\setEEAggbs\, \triangleq \, \bigtimes_{q=1}^Q \Qbareq{\PEEb{q}}  , 
   \end{eqnarray}
   a strategy profile $\QaggbNE \in \fullsetEEAggb$ is a NE iff
    \begin{equation}
     \QbNE{q} = \br{\XbNE{q}}_{\Qbareq{\PEEbNE{q}}}\quad \forall q \in \Omega.
        \label{eq:proj_fixedPoint}
    \end{equation} 
    The above equation states that the strategy profile must be a BR to the other strategies. Moreover, thanks to the projection solution, the BR obviously belongs to the strategy-dependant sets. Such NE can be analyzed in several ways, which include their existence, uniqueness and the design of decentralized algorithms \cite{stupia_power_2015,scutari_competitive_2008,scutari_optimal_2007,scutari_optimal_2008,scutari2009mimo}. These three characteristics are discussed below. 
    \subsection{Existence}
    The existence of the NE follows from the original game definition $\Ge$.
    \begin{proposition}[Existence of NE]
$\Ge$ admits a pure NE.
\end{proposition}
\begin{proof} 
From \cite[Theorem 1.2]{existencepureNE}, a game admits a pure NE if the strategy sets are non-empty compact convex and if the payoff functions are continuous quasi-concave. The sets $\Qeq{P_q}$ satisfying these properties and the quasi-concavity of the objective functions being proven in \cite{xu_energy_2013}, the existence of pure NE follows.
\end{proof}
\subsection{Uniqueness}
Several frameworks can be considered for the analysis of the uniqueness of the NE. As considered in \cite{scutari2010convex,pang_design_2010}, one can derive (Quasi)-Variational Inequalities (see \cite{facchinei_finite-dimensional_2004}) from the optimality conditions characterizing the BR of the users. Another lead, investigated in   \cite{scutari_competitive_2008,miao2011distributed,ren2010distributed,BacciEECompetitive,etkin2007spectrum,scutari_optimal_2008} is to derive conditions under which the BR mapping is a contraction, as in that case it can be proven a unique fixed-point exists. The uniqueness conditions obtained through the two frameworks remarkably coincide when users optimize their rate. Determining whether the two frameworks can lead to the same uniqueness conditions for the EE game remains however an open question which is dealt with in the following by considering the two approaches. Moreover, we focus on the case of square full rank direct channel matrices, while the extension to general matrices is discussed in \Cref{sub:fullcolrank}.
~~\\
\subsubsection{Quasi-Variational Inequality Framework}

Such a framework has been used in \cite{stupia_power_2015} for an OFDM network. Following its lead, we reformulate the problem of finding a NE into an equivalent inequality. Yet, differently from \cite{stupia_power_2015} where the optimality conditions are derived from the original rate function, the QVI we consider here builds on the interpretation of the BR as a projection, this leading to a significant improvement of the uniqueness conditions. Defining
 \begin{equation}
        \Faggb \triangleq \begin{bmatrix} \Fb{1} \\ \vdots \\ \Fb{Q}\end{bmatrix}\triangleq \begin{bmatrix} \Qb{1}-\Xb{1} \\ \vdots \\ \Qb{Q}-\Xb{Q}\end{bmatrix}, 
    \end{equation}
    the following proposition gives a global characterization of the NE through the inner product of Hermitian matrices, i.e. the trace of their product.
\begin{proposition} [QVI reformulation] A strategy profile $\QaggbNE \in \setEEAggbNE$ is a NE of $\Geb$ iff it respects the following quasi-variational inequality $\QVI$:
   \begin{equation}
        \Tr{\FaggbNEH\pa{\Qaggb^{'}-\QaggbNE}}\geq 0 \quad \forall \Qaggb^{'} \in \setEEAggbNE .
        \label{eq:fullQVI}
    \end{equation}
    \end{proposition}
    \begin{proof} 
    From \eqref{eq:EE_projection}, the BR of the $q$th player arises from the following optimization problem:
    \begin{align}
        \BREEb{q} = \argmin{\mathbf{Q}\in\Qbareq{\PEEb{q}}} \frac{1}{2}\norm{\mathbf{Q}-\Xb{q}}{F}^2.
    \end{align}
    From the so-called minimum principle \cite[Section 4.2.3]{boyd_convex_2004}, a matrix $\Qb{q}$ is hence a BR iff  
    \begin{align}
        \Tr{\Fb{q}\pa{\Qb{q}^{'}-\Qb{q}}}\geq 0 \quad \forall \Qb{q}^{'} \in \Qbareq{\PEEb{q}}.
        \label{eq:QVI_indiv}
    \end{align}
    From the definition of $\Faggb$, as all $\Fb{q}$ are Hermitian, summing the $Q$ individual QVIs, \eqref{eq:fullQVI} is obtained. In the reverse direction, letting $ \Qb{r}^{'} = \Qb{r}^{\text{NE}}$ $\forall r \neq q$, \eqref{eq:QVI_indiv} is obtained from \eqref{eq:fullQVI}, and therefore the equivalence follows.
    \end{proof}
When the direct channel matrices are square full rank (and thus invertible), from the definition of $\Xb{q}$ and $\Rmb{q}$, $\Fb{q}$ can be simplified into a linear mapping as 
\begin{align}
    \Fb{q} &= 
    \Hb{qq}^{-1}\mathbf{R}_{n_q}\Hb{qq}^{-H} + 
     \sum_{r=1}^Q\Hb{qq}^{-1}\Hb{qr}\Qb{r}\Hb{qr}^H\Hb{qq}^{-H}.
\end{align}
This linearity intuitively makes the analysis easier. The properties of this QVI mapping $\Faggb$ and the set-valued function $\setEEAggbs$ are essential to the understanding of the behaviour of the QVI, and are analyzed in Appendix \ref{app:QVI_properties}. Building on these properties (Lipschitz continuity and strong monotonicity of $\Faggb$, as well as bounded variation rate of $\setEEAggbs$), we obtain conditions on the uniqueness of the NE of $\Ge$, as stated in \Cref{prop:EEgame_unicty_squarefullrank}. These conditions depend on the nonnegative interference matrix $\Sc$, defined as \begin{equation}
            \br{\Sc}_{qr} \triangleq\left\{\begin{array}{ll}
                  \sigmsq{\Hb{qq}^{-1}\Hb{qr}} & \text{if } r\neq q,   \\
                 0 &\text{otherwise}.
            \end{array}\right. 
                 \label{eq:Sdefinition}
        \end{equation}
Intuitively, the element $\br{\Sc}_{qr}$ will be large if the interference from the $r$-th emitter to the $q$-th receiver is significant with regards to the power of the direct link (between the $q$-th emitter and the same receiver).
\begin{proposition}[NE uniqueness for square nonsingular channel matrices - QVI analysis] If all the channel matrices $\Hb{qq}$ are square full rank, and if the power mapping verifies
\begin{equation}
    \norm{\PEEaggb - \PEEaggbp}{2} < \frac{1-\sr{\mathbf{\Sc}^{s}}}{\sigma_{\max}\pa{\mathbf{I}_Q+\mathbf{\Sc}}} \norm{\Qaggb-\Qaggb^{'}}{F}, 
    \label{eq:EEgame_unicty_squarefullrank}
\end{equation}
$\forall \Qaggb, \Qaggb^{'} \in \fullsetEEAggb$, with $\Sc$ defined in \eqref{eq:Sdefinition}, then $\Ge$ admits a unique NE.
\label{prop:EEgame_unicty_squarefullrank}
\end{proposition}
\begin{proof}
The proof follows from the QVI uniqueness conditions of \cite[Theorem 4 and Corollary 2]{nesterov_solving_2007}, corresponding to the three above lemmas. These conditions being derived in case of a real vector spaces, we prove in Appendix \ref{app:exNesterov} they also hold for complex Hermitian matrices. 
\end{proof}

The above proposition is made off two conditions:
\begin{itemize}
    \item  First, $\sr{\mathbf{\Sc}^{s}}$ must be smaller than 1, which is reminiscent of the uniqueness conditions in case of users maximizing their rate \cite{scutari_competitive_2008}, at the difference that in their case the spectral radius of the non-symmetric matrix is considered. This condition limits the interference caused by and received from the other users \cite[Corollary 11]{scutari_mimo_2009}. This interference limit can also be understood in light of the fact that if nonnegative matrices $\mathbf{A},\mathbf{B}$ are such that $\mathbf{A}\geq \mathbf{B}$, then $\sr{\mathbf{A}}\geq\sr{\mathbf{B}}$ \cite[Fact 4.11.18]{bernstein_dennis_s._matrix_2009}. Hence, lower interference levels (i.e. smaller matrix elements) indeed reduce the spectral radius of $\Sc$.
    \item Secondly, the power mapping needs to vary smoothly with regards to the strategy of the other players. This requirement directly comes from the pseudo-game reformulation \eqref{eq:EEgame_full}: this is one of the advantage of this approach. By reformulating the complex EE problem into this pseudo-game, the convergence conditions are decoupled, leading to interpretable convergence conditions. Yet, it must be kept in mind that this second condition is difficult to check, its simplification constituting an interesting direction for future work. 
\end{itemize}
To summarize, the two conditions can be understood intuitively: when the interference level is low, then the games of the different users are nearly decoupled and therefore one unique equilibrium exist due to the fact that each user has a unique BR.

The above approach is similar to the one of \cite{stupia_power_2015}, where an EE game in an OFDM network is analyzed. Nevertheless, one major difference lies in the formulation of the QVI problem: in our case, the QVI has been derived from the projection solution while in \cite{stupia_power_2015}, the rate objective \eqref{eq:EEgame_full} was considered. This enabled us to obtain a linear QVI while a nonlinear one is analyzed in \cite{stupia_power_2015}, which in turn leads to much tighter uniqueness conditions. Specifically, considering the strong monotonicity condition, \Cref{fig:criteria_comparison} compares our condition particularized to an OFDM network (in green) to the one of \cite{stupia_power_2015} (in blue), as well as the condition obtained in \cite{scutari_competitive_2008} for the rate objective (in orange): $\sr{\Sc}<1$. The curves represent the probability of a satisfied criterion for different SNRs and SIRs, averaged on 2000 i.i.d. realizations: for each SNR and SIR, $2000$ OFDM channel matrices have been drawn from complex normal distributions scaled so as to obtain the targeted average SNR and SIR. The figure clearly shows the benefit of obtaining a linear QVI, as our conditions are much more often satisfied\footnote{Writing the element of $\Sc$ in case of diagonal channel matrices (this corresponding to OFDM networks), one can compare and show that our conditions are strictly tighter than those of \cite{stupia_power_2015}.}. Furthermore, one can also notice the loss of performance due to the symmetrization of $\Sc$, which is discussed below. Finally, the condition obtained in \cite{pan2014totally} for EE in MIMO networks can be seen as a condition on the interference receiver by each user, while ours considers the game in its globality. Yet, differently from us, their condition does not depend anymore on the power mapping and is therefore easier to verify.   \begin{figure}
    \centering
    \includegraphics[width = 0.8\columnwidth]{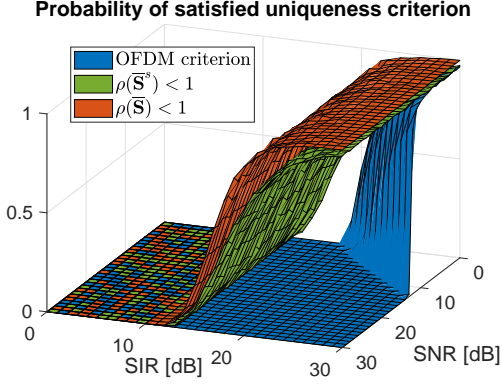}
    \caption{Success probability of the different criteria as a function of the SNR and SIR, averaged on $2000$ i.i.d. realizations. In this graph, higher is better.}
    \label{fig:criteria_comparison}
\end{figure}
~~\\
\subsubsection{Contraction Framework}
In order to analyze the uniqueness problem, the framework of contractive mappings can also be considered, as it has been successfully used in case of users maximizing their rate \cite{scutari_mimo_2009}.  This framework is based on the property that contractive mappings admit a unique fixed-point \cite[Prop. 1.1.a]{bertsekas_parallel_1989}. To that aim, defining the BR mapping as
\begin{equation}
    \BREEAggb = \begin{bmatrix} \BREEb{1} \\ \vdots \\\BREEb{Q} \end{bmatrix},
    \label{eq:def_BREE}
\end{equation}
it appears $\Geb$ admits a unique NE if $\forall \Qaggb, \Qaggb^{'} \in \fullsetEEAggb$,
\begin{align}
    \norm{\BREEAggb - \BREEAggbp}{} \leq \beta \norm{\Qaggb-\Qaggb^{'}}{},
    \label{eq:contraction_EE}
\end{align}
with $\beta <1$ for some choice of norm. Expressing the BR as a projection, the above equation can also be expressed as 
\begin{equation}
    \norm{\br{\Xaggb}_{\setEEAggb} - \br{\Xaggbp}_{\setEEAggbp}}{} \leq \beta \norm{\Qaggb-\Qaggb^{'}}{}.
    \label{eq:contraction_EE_v2}
\end{equation}
Graphically, this condition is depicted in \Cref{fig:contraction}. In this figure, the green (resp. orange) solid line corresponds to $\partial\overline{\mathcal{Q}}(\Qaggb)$ (resp. $\partial\overline{\mathcal{Q}}(\Qaggb^{'})$), these sets being characterized by a different admissible power. The BRs are denoted by the green and orange dots on these lines, corresponding to the projection of $\overline{\mathbf{X}}(\Qaggb)$ and $\overline{\mathbf{X}}(\Qaggb^{'})$. The contraction property of a mapping forces the distance between these BR, denoted by a black line, to be smaller than the distance between the original strategies. This distance depends both on the projected matrices and on the admissible power of each set. As for the QVI analysis, these two effects can be decoupled by considering the projection of $\overline{\mathbf{X}}(\Qaggb)$ onto $\partial\overline{\mathcal{Q}}(\Qaggb^{'})$, represented with a blue point in the figure. Specifically, using the triangular inequality, the black distance can be bounded by the sum of a distance depending only on the projected matrices and of a distance depending only on the transmit powers.
\begin{figure}
    \centering
    \scalebox{1}{
    \begin{tikzpicture}

\definecolor{blue1}{rgb}{0.09, 0.5, 0.7}
\definecolor{blue2}{rgb}{0.05, 0.4, 0.7}

\definecolor{blue3}{rgb}{0.00000,0.44700,0.74100}%
\definecolor{orange_scal}{rgb}{0.85000,0.32500,0.09800}%
\definecolor{forestgreen}{rgb}{0.46600,0.67400,0.18800}%
\definecolor{blue0}{rgb}{0,0,0}%
    \tikzmath{\u=1;}
    \tikzmath{\widthaxis=2.5*\u;\width = 2.3*\u; \widthtwo = 1.8*\u; \widthc = 0.2*\u;}
    \draw (-\widthaxis,0) -- (\widthaxis,0);
    \draw (0,-\widthaxis) -- (0,\widthaxis);

    \draw [ultra thick,orange_scal] (0, \width) -- (\width,0);
     \draw [ultra thick,forestgreen] (0, \widthtwo) -- (\widthtwo,0);
     \node[] at (-0.7*\u,\width) {\textcolor{orange_scal}{\normalsize{\textbf{$\PEEaggbp$}}}};
     \node[] at (-0.7*\u,\widthtwo-0.2*\u) {\textcolor{forestgreen}{\normalsize{\textbf{$\PEEaggb$}}}};

\draw[draw = forestgreen, fill = forestgreen] (-1.5*\u,-2*\u)  circle (\widthc/2);
\node[] at (-2.2*\u,-1.9*\u) {\textcolor{forestgreen}{\normalsize{\textbf{$\Xaggb$}}}};
\draw[draw = orange_scal, fill = orange_scal] (-0.8*\u,-0.5*\u)  circle (\widthc/2);
\node[] at (-1.5*\u,-0.5*\u) {\textcolor{orange_scal}{\normalsize{\textbf{$\Xaggbp$}}}};

\tikzmath{\ponex = 1.15*\u;\poney=\widthtwo-\ponex;}
\draw[draw = forestgreen, fill = forestgreen] (\ponex,\poney)  circle (\widthc/2);
\draw [ultra thick,forestgreen, dashed] (-1.5*\u,-2*\u) -- (\ponex,\poney);
\tikzmath{\ptwox =0.97*\u;\ptwoy=\width-\ptwox;}
\draw[draw = orange_scal, fill = orange_scal] (\ptwox,\ptwoy) circle (\widthc/2);
\draw [ultra thick,orange_scal, dashed](-0.8*\u,-0.5*\u) -- (\ptwox,\ptwoy);

\node[] at (\ponex-0.8*\u,\poney) {\textcolor{forestgreen}{\normalsize{$\BREEAggb$}}};
\node[] at (\ptwox+0.8*\u,\ptwoy+0.1*\u) {\textcolor{orange_scal}{\normalsize{$\BREEAggbp$}}};

\draw [ultra thick,blue0] (\ponex,\poney) -- (\ptwox,\ptwoy);

\tikzmath{\pthreex =1.41*\u;\pthreey=\width-\pthreex;}
\draw[draw = blue3, fill = blue3] (\pthreex,\pthreey) circle (\widthc/2);
\node[] at (\pthreex+1.3*\u,\pthreey-0.1*\u) {\textcolor{blue3}{\normalsize{\textbf{$\br{\Xaggb}_{\setEEAggbp}$}}}};
\draw [ultra thick,blue3, dashed](\pthreex,\pthreey) -- (\ponex,\poney);


\end{tikzpicture}}
    \caption{Graphical representation of \eqref{eq:contraction_EE_v2}.}
    \label{fig:contraction}
\end{figure}
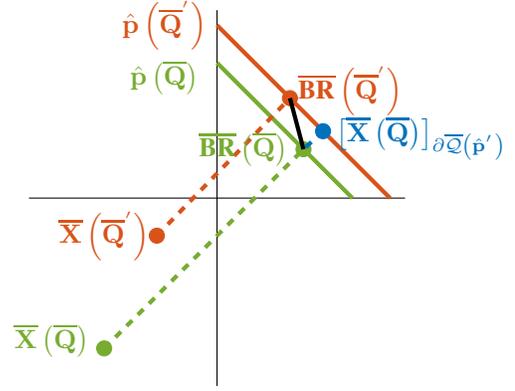

Before obtaining the uniqueness results, it should be noted that the way distances are measured can be chosen to get as tight conditions as possible. To that aim, similarly to \cite{scutari_competitive_2008}, the weighted block-maximum norm and the corresponding 2-norm for vectors are introduced as 
\begin{equation}
    \normw{\Qaggb}{F} = \max_{q\in \Omega}\frac{\norm{\Qb{q}}{F}}{w_q},\quad \normw{\mathbf{u}}{\infty} = \max_{q\in \Omega}\frac{\abs{u_q}}{w_q},
\end{equation}

for some choice of strictly positive $\mathbf{w}$. 
\begin{proposition}[NE uniqueness for square nonsingular channel matrices - contraction analysis]
If all the channel matrices $\Hb{qq}$ are square full rank, and if the power mapping verifies  
\begin{align}
    \normw{\PEEaggb - \PEEaggbp}{\infty} < \pa{1-\sr{\mathbf{\Sc}}} \normw{\Qaggb-\Qaggb^{'}}{F},
    \label{eq:contraction_cond}
    \end{align}
$\forall \Qaggb, \Qaggb^{'} \in \fullsetEEAggb$, with $\Sc$ being defined in \eqref{eq:Sdefinition} and $\mathbf{w}$ its right Perron eigenvector, then $\Ge$ admits a unique NE.
     \label{prop:EE_unicity_general_contraction}
\end{proposition}
\begin{proof}
From the triangular inequality, building on the above description of \Cref{fig:contraction},
\begin{align}
    & \normw{\br{\Xaggb}_{\setEEAggb} - \br{\Xaggbp}_{\setEEAggbp}}{F} \nonumber \\
    & \leq \normw{\br{\Xaggb}_{\setEEAggb} - \br{\Xaggb}_{\setEEAggbp}}{F} \nonumber\\ & \qquad  + \normw{\br{\Xaggb}_{\setEEAggbp}- \br{\Xaggbp}_{\setEEAggbp}}{F}. 
    \label{eq:EE_sum_dist}
\end{align}
In the above equation, the first term is the projection of the same matrix onto two different sets which, from \Cref{lemma:smoothpower}, can be bounded as
\begin{align}
    &\normw{\br{\Xaggb}_{\setEEAggb}- \br{\Xaggb}_{\setEEAggbp}}{F} \nonumber \\&\qquad \qquad\qquad\qquad\qquad\leq \normw{\PEEaggb-\PEEaggbp}{\infty}.
    \label{eq:ineqsmoothpower}
\end{align}
The second term is the projection of two different matrices on the same set. This term is in fact similar to what one would get from the rate game following a contraction approach, and it has been proven in \cite[Theorem 7]{scutari_mimo_2009} to be bounded by  
\begin{align}
    &\normw{\br{\Xaggb}_{\setEEAggbp} - \br{\Xaggbp}_{\setEEAggbp}}{F} \nonumber \\&\qquad \qquad \qquad \qquad \qquad \qquad\leq \sr{\Sc} \normw{\Qaggb-\Qaggb^{'}}{F},
    \label{eq:VI_contration_sr}
\end{align}
with $\mathbf{w}$ the right Perron eigenvector of the nonnegative matrix $\Sc$ \cite[Prop. 6.6]{bertsekas_parallel_1989}.
Using the two above bounds, a sufficient condition for the BR to be a contraction is 
\begin{align}
    &\sr{\Sc} \normw{\Qaggb-\Qaggb^{'}}{F} + \normw{\PEEaggb-\PEEaggbp}{\infty} \nonumber \\& \qquad \qquad \qquad \qquad \qquad \qquad \qquad  < \normw{\Qaggb-\Qaggb^{'}}{F}, 
    \label{eq:finalineq}
\end{align}
which concludes the proof.
\end{proof}

As for the QVI conditions, the above proposition can be interpreted as a limit on the interference level, as well as a limit on the sensitivity of the power mapping. It should furthermore be noted that the condition on the spectral radius is the same as state-of-the-art uniqueness conditions for users maximizing their rate. Moreover, the choice of $\mathbf{w}$ enables to obtain tight conditions for \eqref{eq:VI_contration_sr} but nothing says this choice is optimal for \eqref{eq:ineqsmoothpower}. A more cautious approach would be to chose $\mathbf{w}$ in order to obtain a tight global inequality \eqref{eq:finalineq}. Nevertheless, due to the complexity of the power mapping, this turns out to be very difficult.
~~\\
\subsubsection{Comparison between the QVI and Contraction Framework}
Comparing the QVI approach and the contraction approach, several comments can be made. First, one can notice the spectral radius is taken on $\Sc$ directly, this leading to tighter uniqueness conditions as \Cref{fig:criteria_comparison} shows, and as proven in \cite{schwenk_tight_1986}. Secondly, as $\sigma_{\max}\pa{\mathbf{I}_Q + \Sc} >1 $ \cite[Fact 4.11.18]{bernstein_dennis_s._matrix_2009}, the Lipschitz continuity requirement of the QVI approach appears to be more demanding than the contraction approach. Obviously, the effectiveness of the QVI approach depends on the bounds' quality of \Cref{lemma:Lipschitz}. To that aim, the following lemma proves that the Lipschitz constant cannot be equal to 1, and thus that the QVI approach cannot compete with the contraction approach, at least when considering the uniqueness theorem of \cite{nesterov_solving_2007}.  

\begin{lemma}
    The Lispchitz constant of the QVI mapping is at least greater than $\sqrt{Q}$. 
\end{lemma}
\begin{proof}
Consider a situation where all channel matrices are the identity matrix, and consider $\Qb{q} = \Qb{q}^{'} \quad \forall \, q \neq o$. These channel matrices and strategies lead to the following, as $\Eb{o}$ is the only nonzero matrix.
\begin{align}
    \norm{\Faggb-\Faggbp}{F}^2 &= \sum_{q=1}^Q\norm{\Fb{q}-\Fbp{q}}{F}^2,\nonumber \\
    & = \sum_{q=1}^Q\norm{\sum_{r=1}^Q\Eb{r}}{F}^2,\\
    &= Q \norm{\Eb{o}}{F}^2 = Q\norm{\Qaggb-\Qaggb^{'}}{F}^2,
\end{align}
concluding the proof.
\end{proof}

Therefore, contrarily to games in which users maximize their rate, the QVI framework is not as efficient as the contraction approach. Whether general QVI uniqueness results can be made tighter remains an open question, left for future work.
~~\\
\subsubsection{Back to the Original Game}
The unicity conditions, whatever the considered framework, depend on the interference matrix $\Sc$, which depends itself on the modified channel matrices $\Hb{qr}$. In the following, we express the elements of $\Sc$ with the original channel matrices. We only consider full-row rank channel matrices, as row-rank deficient matrices lead to rectangular modified channel matrices which have not been studied up to this point. 
\begin{corollary}
        \label{cor:unicityCond_fullrowrank}
        If $\mathbf{H}_{qq}$ is full row-rank, $\br{\Sc}_{qr}$ can be defined as
        \begin{equation}
             \br{\Sc}_{qr} =\left\{\begin{array}{ll} \sigmsq{\mathbf{H}_{qq}^{\sharp}\mathbf{H}_{qr}\mathbf{V}_{r,1}}& \text{if } r\neq q,   \\
                 0 &\text{otherwise}, 
            \end{array}\right. 
        \end{equation}
        with $\mathbf{V}_{r,1}$ the right unitary matrix of the compact SVD of $\mathbf{H}_{rr}$.
        \end{corollary}
        \begin{corollary}
        \label{cor:unicityCond_squarefullrank}
         If $\mathbf{H}_{qq}$ is full row rank and if $\mathbf{H}_{rr}$ is square nonsingular, $\br{\Sc}_{qr}$ can be simplified as
        \begin{equation}
             \br{\Sc}_{qr} = \left\{\begin{array}{ll} \sigmsq{\mathbf{H}_{qq}^{\sharp}\mathbf{H}_{qr}}& \text{if } r\neq q,   \\
                 0 &\text{otherwise}.
            \end{array}\right. 
        \end{equation}
        \end{corollary}
        \begin{proof}
        The Moore-Penrose inverse of $\mathbf{H}_{qq}$ can be written in terms of its SVD as \cite[Section 6.1]{bernstein_dennis_s._matrix_2009} $       \mathbf{H}_{qq}^{\sharp} = \mathbf{V}_{q,1}\mathbf{\Sigma}_{qq}^{-1}\mathbf{U}_{q,1}^H$.    Hence, as $\Hb{qq}=\mathbf{H}_{qq}\mathbf{V}_{q,1}$ and from the semi-unitary property of $\mathbf{V}_{q,1}$,
        \begin{align}            &\mathbf{H}_{qq}^{\sharp^H}\mathbf{H}_{qq}^{\sharp} = \mathbf{U}_{q,1}\mathbf{\Sigma}_{qq}^{-1}\mathbf{V}_{q,1}^H\mathbf{V}_{q,1}\mathbf{\Sigma}_{qq}^{-1}\mathbf{U}_{q,1}^H \nonumber \\ &\qquad \qquad\quad = \mathbf{U}_{q,1}\mathbf{\Sigma}_{qq}^{-1}\mathbf{\Sigma}_{qq}^{-1}\mathbf{U}_{q,1}^H = \Hb{qq}^{-H}\Hb{qq}^{-1}.  
        \end{align}
        Hence, from the definition of $\Hb{qr} = \mathbf{H}_{qr}\mathbf{V}_{r,1}$,
        \begin{align}
         \sigmsq{\Hb{qq}^{-1}\Hb{qr}} &= \sr{\Hb{qr}^H\Hb{qq}^{-H}\Hb{qq}^{-1}\Hb{qr}},\nonumber \\
         & = \sr{\mathbf{V}_{r,1}^H\mathbf{H}_{qr}^H\mathbf{H}_{qq}^{\sharp^H}\mathbf{H}_{qq}^{\sharp}\mathbf{H}_{qr}\mathbf{V}_{r,1}},\\
         &=\sigmsq{\mathbf{H}_{qq}^{\sharp}\mathbf{H}_{qr}\mathbf{V}_{r,1}},\\
         &\leq  \sigmsq{\mathbf{H}_{qq}^{\sharp}\mathbf{H}_{qr}},\label{eq:inequality_singular}
        \end{align}
        proving \Cref{cor:unicityCond_fullrowrank}, the above inequality coming from the fact that the maximum singular value of a product is smaller than the product of the maximum singular values \cite[Corollary 9.3.7]{bernstein_dennis_s._matrix_2009} and from $\sigma_{\max}\pa{\mathbf{V}_{r,1}}=1$. If $\mathbf{H}_{rr}$ is square nonsingular, then $\mathbf{V}_{r,1}$ is square unitary and the equality holds as the singular values are unitarily invariant, completing the proof of \Cref{cor:unicityCond_squarefullrank}. 
        \end{proof}
Intuitively, when channel matrices are full-row rank (i.e. when there is more antennas at the $r$th transmitter than at the $q$th receiver), the $\mathbf{V}_{r,1}$ factor enables to take into account the fact that, even if the $r$th transmitter has a large number of antennas to play with, the interference affects only a small number of antennas at the $q$th transmitter. Comparing the above relations to those obtained in \cite{scutari_mimo_2009} for the rate game, they match in case of full rank square direct channel matrices. However, in case of full-row rank matrices, the authors of \cite{scutari_mimo_2009} lack the $\mathbf{V}_{r,1}$ factor for the elements of $\Sc$, even though this factor, from \eqref{eq:inequality_singular}, leads to smaller matrix elements and therefore to tighter  uniqueness conditions\footnote{In the proof of \cite{scutari_mimo_2009}, one could replace the strategy profiles $\mathbf{Q}_r$ in equation (32) by the solution characterization obtained in \cite[Appendix C]{scutari_mimo_2009}: $\mathbf{V}_{r,1}\mathbf{V}_{r,1}^H\mathbf{Q}_r\mathbf{V}_{r,1}\mathbf{V}_{r,1}^H$. Doing so, the $\mathbf{V}_{r,1}$ factor would appear in the unicity criterion.}. 
~~\\
\subsubsection{Full-Column Rank Matrices}
When the channel matrices are full-column rank, no reverse-order law enables to simplify $\Xb{q}$, leading therefore to a nonlinear mapping $\Faggb$ in the QVI approach, and to a more complex projection in the contraction framework. Similarly to the approach followed in \cite{scutari_mimo_2009}, uniqueness results can be obtained through the use of a mean value theorem for complex matrices. Specifically, as discussed in the following, all the above results hold, yet with a modified interference matrix $\Scb$, defined as
\begin{equation}
           \br{\Scb}_{qr} \triangleq \left\{\begin{array}{ll}
            \underset{\Db\in\setAggb}\max\sigmsq{\Gb{qr}{\Db}} & \text{if } r\neq q,   \\
                 0 &\text{otherwise},
            \end{array}\right.
            \label{eq:definition_Sbar}
        \end{equation}
with \begin{equation}
            \Gb{qr}{\Db} \triangleq \invIntbD{q}\Hb{qq}^H\RmbD{q}\Hb{qr}.
            \label{eq:definitionG}
        \end{equation}
It is important to note that if $\Hb{qq}$ is square nonsingular, then $\overline{\mathbf{G}}_{qr} = \Hb{qq}^{-1}\Hb{qr}$ does not depend on $\Db$ anymore and therefore $\Scb$ boils down to $\Sc$. 
\begin{proposition}[NE uniqueness for full-column rank channel matrices - contraction analysis]
If the power mapping verifies
\begin{align}
    \normw{\PEEaggb - \PEEaggbp}{\infty} < \pa{1-\sr{\mathbf{\Scb}}} \normw{\Qaggb-\Qaggb^{'}}{F},
    \label{eq:contraction_condition_full}
    \end{align}
$\forall \Qaggb, \Qaggb^{'} \in \fullsetEEAggb$, with $\Scb$ defined in \eqref{eq:definition_Sbar} and $\mathbf{w}$ its right Perron eigenvector, then $\Ge$ admits a unique NE.
    
\end{proposition}
\begin{proof}
The proof follows from the fact that \Cref{lemma:smoothpower} does not make any assumption on the rank of the channel matrices, and from \cite[Theorem 7]{scutari_mimo_2009}. Building on these results, the proof of \Cref{prop:EE_unicity_general_contraction} can be easily adapted.
\end{proof}

\label{sub:fullcolrank}

\section{Iterative algorithm}
\label{sec:algo}
We have obtained above conditions on the existence and uniqueness of the NE of $\Ge$. Assuming a unique NE exists, this section describes an asynchronous decentralized algorithm which echoes the totally asynchronous IWFA developed for the MIMO rate game in \cite{scutari_competitive_2008}. To that aim, the set of times at which the $q$th player updates its strategy is denoted by $\mathcal{T}_q$. Moreover, this update is performed given the interference level of the other users which needs to be measured beforehand. The most recent time at which the $q$th user has measured the interference coming from the $r$th user is denoted by $0\leq \tau^q_r(t)\leq t$. Grouping together the strategies measured at the times $\tau^q_r(t)$, the $q$th user updates its strategy based on the strategy profile of the other users $\mathbf{Q}_{-q}^{\pa{\tau^q(t)}}$. We assume that the update and measure times satisfy mild conditions which are verified for any practical wireless system, ensuring that all users update at some point their strategy and that old information is eventually purged from the system \cite{bertsekas_parallel_1989}. With these definitions, the totally asynchronous EE-IWFA is presented in \Cref{algo:IWFA}.

\begin{algorithm}
\caption{EE-IWFA - $q$th user}
\begin{algorithmic}
\State Set $t=0$, $\mathbf{Q}_q^{\pa{0}} \in \mathcal{Q}\pa{\mathbf{p}}$.
    \For{$t = 0 \rightarrow \infty$}
        $$\mathbf{Q}_q^{(t+1)} = \left\{\begin{array}{cc}
            \textbf{\text{BR}}_{q}\pa{\mathbf{Q}_{-q}^{\pa{\tau^q(t)}}} &  \text{if} \: t \in \mathcal{T}_q,\\
             \mathbf{Q}_q^{(t)}&  \text{otherwise.}
        \end{array}\right.$$
    \EndFor
\label{algo:IWFA}
\end{algorithmic}
\end{algorithm}

\begin{proposition}[Convergence of the EE-IWFA]
    If condition \eqref{eq:contraction_condition_full} holds, then the EE-IWFA converges linearly to the unique NE of $\Ge$.
\end{proposition}
\begin{proof}
    The proof follows directly from the fact that the BR is a contraction, and can be found in \cite[Appendix D]{scutari2009mimo}. 
\end{proof}
Note that in order to compute their BR, users need to know neither the cross-channel matrices, nor the exact strategy of the other users. The only needed quantities are the direct channel matrix as well as the interference level $\Ri{q}$ (one can check that $\Ri{q} = \Rib{q}$). Hence, no information needs to be exchanged between the TRPs and the algorithm can be run in a fully decentralized way. Varying the sets $\mathcal{T}_q$, the above algorithm can either be sequential, synchronous, or asynchronous. Comparing this decentralized algorithm to the one of \cite{stupia_power_2015}, the EE-IWFA benefits from the fact that it can be run in a fully asynchronous way, at that it is in line with the BR of the players. Numerical results are provided in the following, both for synchronous and asynchronous updates. We stress out that these numerical results are obtained for specific sets of channel matrices, and that a systematic evaluation of the EE-IWFA performances are an interesting direction for future work. Yet, these examples enable to grasp the main feature of this decentralized algorithm. We consider a scenario similar to \cite{stupia_power_2015}, with $8$ players having $n_{T_q} = n_{R_q} = r_q = 4$ (for simplicity, we consider full rank square channel matrices.). The maximum power corresponds to the number of antennas, so that a uniform power allocation assigns $1$ unit of power on each of the MIMO directions ($P_q =4$), while the hardware circuitry power is set to $\Psi_q = 1$. The noise is considered independent on each of the antennas, leading to diagonal $\mathbf{R}_{n_q}$ and such that $\text{SNR}_q = 7\text{dB}$. Each element of the channel matrices is drawn independently from a complex normal distribution with variance $\sigma^2_{qr}$ for matrix $\mathbf{H}_{qr}$, the variance being chosen so as to obtain the target SIR as \begin{align}\sigma^2_{qr} =\left\{\begin{array}{cc}
        1  &  \text{if} \quad r =  q,\\
        \frac{1}{(Q-1)\text{SIR}} & \text{otherwise}.
    \end{array}\right.\end{align}Finally, the threshold of the iterative Dinkelbach method is set to $\epsilon = 10^{-9}$.
\begin{figure}[ht]
\centering
\scalebox{0.5}{\input{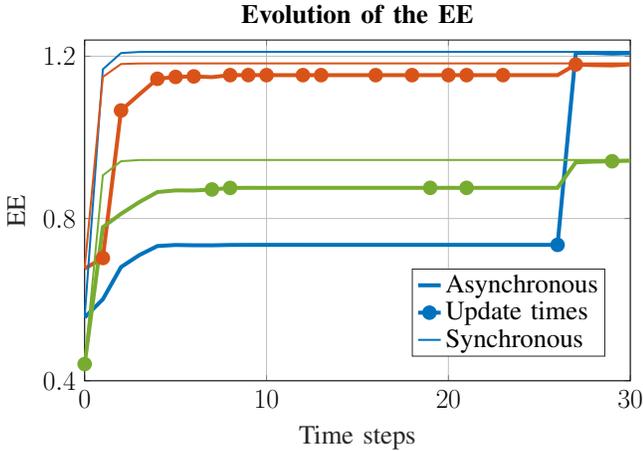}}
\caption{Evolution of the EE of three of the players for a SIR of $0$dB.}
\label{fig:rateGE}
\end{figure}
In \Cref{fig:rateGE}, the EE of three of the users is shown as a function of the time steps, for a moderate SIR. The synchronous iterations are shown with thin lines while the asynchronous iterations are represented with thick lines. The dots represent the update times, which differ from user to user. From this graph, one can observe that both algorithms converge to the same NE. Interestingly, even players allowed to update their strategy far less often than others converge in few iterations. As far as the converge rate is concerned, \Cref{fig:convergence_IWFA} shows the convergence rate of the EE-IWFA, measured as the block maximum norm of the difference between consecutive strategies \cite[Appendix D]{scutari_mimo_2009}, for different SIRs. For a high to moderate SIR, it appears the convergence is linear (as expected), and limited by the accuracy of the Dinkelbach method. More interestingly, this example shows that even when the unicity and convergence condition is not fulfilled ($\sr{\Sc}>1$), EE-IWFA may still converge as all the above analyses consider worst case guarantees. Eventually, for very low SIRs and thus very high spectral radii, the method does not converge. 
\begin{figure}[ht]
  \centering
  \scalebox{0.5}{
%
%
\definecolor{mycolor1}{rgb}{0.00000,0.44700,0.74100}%
\definecolor{mycolor2}{rgb}{0.85000,0.32500,0.09800}%
\definecolor{mycolor3}{rgb}{0.92900,0.69400,0.12500}%
\definecolor{mycolor4}{rgb}{0.49400,0.18400,0.55600}%
\definecolor{mycolor5}{rgb}{0.46600,0.67400,0.18800}%
\definecolor{mycolor6}{rgb}{0.30100,0.74500,0.93300}%
\definecolor{mycolor7}{rgb}{0.63500,0.07800,0.18400}%

\begin{tikzpicture}
\huge
\begin{axis}[%
width=0.6\textwidth,
height=0.6\textwidth,
at={(0.892in,0.646in)},
scale only axis,
xmin=1,
xmax=30,
xlabel style={font=\color{white!15!black}},
xlabel={Time steps},
ymode=log,
ymin=1e-13,
ymax=1e2,
yminorticks=false,
ytick={1e-9,1e-6,1e-3,1,1e2},
yticklabels={{$\epsilon$},{$10^{-6}$},{$10^{-3}$},{$10^{0}$},{$10^{2}$}},
ylabel style={font=\color{white!15!black}},
ylabel={Norm difference},
axis background/.style={fill=white},
title style={font=\bfseries},
title={Convergence analysis},
xmajorgrids,
ymajorgrids
]

\addplot [color=mycolor2, line width=4.0pt, forget plot]
  table[row sep=crcr]{%
1	53.8779076387108\\
2	0.834956830934695\\
3	0.00427272784485808\\
4	3.97547206006827e-05\\
5	3.33437419775465e-07\\
6	2.40408377794344e-09\\
7	7.01747285437845e-11\\
8	4.61595006662368e-11\\
9	4.79380131946968e-11\\
10	1.22182790399556e-09\\
11	3.09372402394547e-09\\
12	3.1260453514248e-09\\
13	8.47644211479453e-12\\
14	2.6253876215605e-10\\
15	2.6253457144568e-10\\
16	1.85775733160133e-10\\
17	1.03385066906472e-10\\
18	3.015089676067e-10\\
19	3.03692336317974e-10\\
20	2.75594212607552e-10\\
21	2.76398977987857e-10\\
22	1.55742149737224e-11\\
23	1.02565684968517e-11\\
24	4.37713202688432e-11\\
25	4.37912974492179e-11\\
26	1.37659059753653e-09\\
27	1.30593203074393e-09\\
28	1.01981204806441e-10\\
29	1.01129204110099e-10\\
30	2.53947952734534e-10\\
};

\addplot [color=mycolor1, line width=4.0pt, forget plot]
  table[row sep=crcr]{%
1	41.3489629393179\\
2	4.27693178203504\\
3	0.759335367112384\\
4	0.0667341879558181\\
5	0.00548545073861869\\
6	0.000581932603620935\\
7	6.86276437863985e-05\\
8	9.75634086077521e-06\\
9	9.42858414146696e-07\\
10	9.38802393238077e-08\\
11	9.68169740867078e-09\\
12	2.39216151478834e-09\\
13	2.40427264183501e-09\\
14	8.1587776258164e-10\\
15	4.920950973926e-10\\
16	4.74127338403811e-10\\
17	1.8792610908065e-10\\
18	2.14094791363109e-09\\
19	2.14955367179176e-09\\
20	1.25349435477503e-10\\
21	7.5141728925206e-09\\
22	7.41321328818738e-09\\
23	8.10886617093617e-10\\
24	3.55801135671647e-10\\
25	3.80115384844211e-10\\
26	4.80139488453552e-09\\
27	5.78952750476043e-09\\
28	5.76121985724746e-09\\
29	1.60799934600583e-09\\
30	1.73700424400091e-09\\
};

\addplot [color=mycolor5, line width=4pt, forget plot]
  table[row sep=crcr]{%
1	32.8670921008595\\
2	45.7174733336741\\
3	26.8449998744014\\
4	27.0824863314817\\
5	36.0377680319322\\
6	21.4882477163578\\
7	9.59675968713632\\
8	26.8838770582825\\
9	34.1120579774128\\
10	29.3913546791787\\
11	20.1933097457829\\
12	23.2631820348537\\
13	34.3334039080979\\
14	28.9210294821054\\
15	14.8231107292535\\
16	34.2185995470558\\
17	38.638369114261\\
18	20.7372508887709\\
19	16.1237792961638\\
20	25.8389343213241\\
21	32.7895334253228\\
22	15.694810331475\\
23	20.3579120604919\\
24	17.7096928103752\\
25	18.7349426535923\\
26	6.32087736883492\\
27	18.924137203872\\
28	11.0974711105678\\
29	12.5736822644402\\
30	9.06164855440729\\
};
\end{axis}
\node[align = center] at (4.1,2.6) {\textcolor{mycolor2}{$13dB$} \\\textcolor{mycolor2}{$\sr{\Sc}=0.8$}};
\node[align = center] at (7.8,7) {\textcolor{mycolor1}{$0dB$}\\\textcolor{mycolor1}{$\sr{\Sc}=15.6$}};
\node[align = center] at (7.5,10.8) {\textcolor{mycolor5}{$-18dB$}\\\textcolor{mycolor5}{$\sr{\Sc}=567$}};

\end{tikzpicture}%
}
\caption{Convergence, measured as the block maximum norm of the difference between consecutive strategies, of the synchronous EE-IWFA for different SIRs.}
\label{fig:convergence_IWFA}
\end{figure}
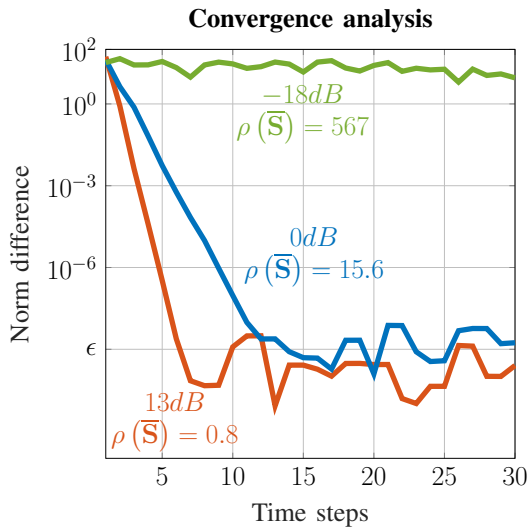%
In \Cref{fig:rateGE_low}, the evolution of the EE of one particular player (coming from the same channel matrices than those of the low SIR curve of \Cref{fig:convergence_IWFA}) is represented. In this figure, it appears the EE-IWFA is stuck in periodic updates. Interestingly, the asynchronous IWFA is more erratic, and starts converging after $t=450$. This can be explained by the fact that as update times are random, asynchronous IWFA explores more strategies, as for example around $t=230$. Finally, one should keep in mind that the above figures display results corresponding to particular channel matrices, and that therefore no rule allows to link the SIR to a convergence guarantee. 

\begin{figure}[ht]
\centering
\scalebox{0.53}{\input{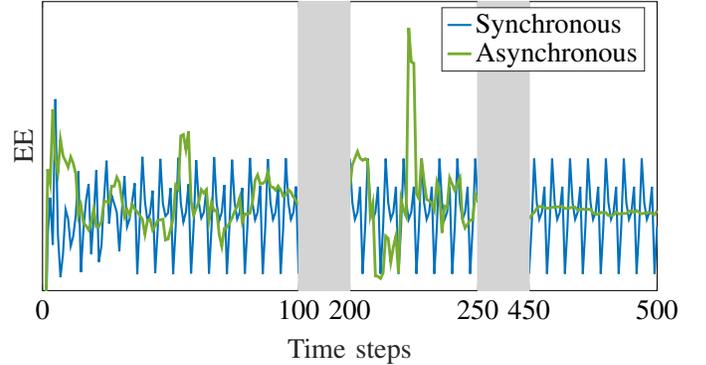}}
\caption{Evolution of the EE for a SIR of $-18$dB. For the sake of clarity, the update times of the asynchronous algorithm are not represented.}
\label{fig:rateGE_low}
\end{figure}

\section{Conclusion}
In this work, we have obtained conditions on the feasibility of competitive resource allocation in MIMO networks for TRPs maximizing their EE. We have proven that if the interference is limited, and if the optimal power varies slowly with regards to the players' strategies, then the EE-IWFA converges to the unique NE in a fully decentralized asynchronous way. The obtained conditions, when particularized to an OFDM network, outperform those of existing literature. At a higher level, two mathematical frameworks widely used to study wireless resource allocations, namely the QVI and the contraction frameworks, have been compared. We have shown that the contraction approach outplays the QVI approach for the MIMO EE resource allocation. Future work includes an in-depth study of the power mapping, as well as the analysis of other QVI uniqueness tools. 



%

\appendices
\section{Extension of uniqueness conditions to complex Hermitian matrices}
\label{app:exNesterov}
Considering a QVI defined on a real vector space, unicity conditions are provided in \cite[Theorem 4 and Corollary 2]{nesterov_solving_2007}. Yet, in our case, we work with complex Hermitian matrices and therefore the above results do not hold straightforwardly. Fortunately, similarly to the bijection between complex numbers and $2\times 2$ real matrices:  
\begin{equation}
    z = a +bj \quad \text{ and } \quad \mathbf{Z} = \frac{1}{2}\begin{bmatrix} a & -b \\ b & a \end{bmatrix},
\end{equation}
this lemma can be generalized to $M\times N$ complex matrices through their bijection with $2M\times 2N$ real matrices.
    Considering $\mathbf{Z} \in \C{M}{N}$, whose elements are denoted by $z_{mn} = a_{mn}+j b_{mn}$ with $a_{mn},b_{mn} \in \mathbb{R}$, the following bijection is defined:
    \begin{footnotesize}
    \begin{align}
    &\mathbf{Z} = \begin{pmatrix}
    z_{11} &  \hdots & z_{1N}\\
    \vdots & \ddots&\vdots\\
    z_{M1} & \hdots& z_{MN}
    \end{pmatrix} \nonumber \\ & \qquad  \iff \mathbf{Z}^{\dagger} = \frac{1}{2} \begin{pmatrix}
    a_{11} & -b_{11}& \hdots & a_{1N} & -b_{1N}\\
    b_{11} & a_{11} &\hdots & b_{1N} & a_{1N}\\
    \vdots & \vdots&\ddots& \vdots&\vdots\\
    a_{M1} &-b_{M1} & \hdots& a_{MN} & -b_{MN}\\
    b_{M1} &a_{M1} & \hdots& b_{MN} & a_{MN}\\
    \end{pmatrix}.
\end{align}
    \end{footnotesize}
From the bijection definition, the following properties can be verified (properties i), iii) and v) being specific to square matrices):
\begin{enumerate}[i)]
    \item $\mathbf{Z} = \mathbf{Z}^H\iff \mathbf{Z}^{\dagger} = \mathbf{Z}^{\dagger^T}$ (as the diagonal terms of $\mathbf{Z}$ are real);
    \item $\mathbf{Z} = \mathbf{XY} \iff \mathbf{Z}^{\dagger} = \mathbf{X}^{\dagger}\mathbf{Y}^{\dagger}$;
    \item $\Tr{\mathbf{Z}} = \Tr{\mathbf{Z}^{\dagger}}$ (for square matrices);
    \item $\norm{\mathbf{Z}}{F}^2 = \Tr{\mathbf{Z}^H\mathbf{Z}} = \Tr{\mathbf{Z}^{\dagger^T}\mathbf{Z}^{\dagger}} = \norm{\mathbf{Z}^{\dagger}}{F}^2$;
    \item $\mathbf{Z}=\mathbf{U\Lambda U}^H \iff \mathbf{Z}^{\dagger}=\mathbf{U}^{\dagger}\mathbf{\Lambda}^{\dagger} \mathbf{U}^{\dagger^T}$ with $\mathbf{\Lambda}^{\dagger}$ diagonal and $\mathbf{U}^{\dagger}\mathbf{U}^{\dagger^T} = \mathbf{I}_{2N}$. Hence, $\mathbf{Z}$ and $\mathbf{Z}^{\dagger}$ have the same eigenvalues but with double multiplicity.
\end{enumerate}
Thanks to the above properties, $\QVI$ can equivalently be defined with real matrices of double size. Hence, uniqueness conditions of \cite{nesterov_solving_2007} can be applied to this equivalent QVI. Again, the three conditions of the lemma can be translated back to the Hermitian matrices through the bijection.

\section{Properties of the EE QVI}
\label{app:QVI_properties}
This appendix focuses on the properties of $\QVI$ which are instrumental to analyze the NE of the EE game. More precisely, we obtain the Lipschitz constant $L$ and the strong monotonicity constant $\mu$ of $\Faggb$ as a function of the interference matrix \eqref{eq:Sdefinition}, respectively in \Cref{lemma:Lipschitz} and \Cref{lemma:monoton}. We also obtain a bound on the variation rate of the strategy set $\setEEAggbs$ in \Cref{lemma:smoothpower}. 
\begin{lemma}
If all the channel matrices $\Hb{qq}$ are square full rank, the operator $\Faggb$ is Lipschitz continuous with constant $L = \sigma_{\max}\pa{\mathbf{I}_Q + \Sc}$: $\forall \Qaggb,\Qaggb^{'}\in \fullsetEEAggb$,
    \begin{equation}
     \norm{\Faggb-\Faggbp}{F}\leq L\norm{\Qaggb-\Qaggb^{'}}{F}. 
     \label{eq:Lipschitz}
     \end{equation}
     \label{lemma:Lipschitz}
\end{lemma}
\begin{proof}
Defining $\Eb{q} \triangleq \Qb{q}- \Qb{q}^{'}$,
 \begin{align} &\norm{\Fb{q}-\Fbp{q}}{F} \nonumber\\
         &\leq \norm{\Eb{q}}{F} + \sum_{r\neq q}\norm{\Hb{qq}^{-1}\Hb{qr}\Eb{r}\Hb{qr}^H\Hb{qq}^{-H}}{F},\label{eq:EE_triangleeq} \\
         &\leq \norm{\Eb{q}}{F} + \sum_{r\neq q}\sr{\Hb{qr}^H\Hb{qq}^{-H}\Hb{qq}^{-1}\Hb{qr}}\norm{\Eb{r}}{F},\label{eq:EE_norm_ineq}
    \end{align}
    where \eqref{eq:EE_triangleeq} follows from the triangular inequality and \eqref{eq:EE_norm_ineq} from \cite[Prop. 8.4.13]{bernstein_dennis_s._matrix_2009} and the cyclic property of the trace. Letting $\mathbf{f} \triangleq \begin{bmatrix} f_1 \hdots f_Q\end{bmatrix}^T$ with $f_q = \norm{\Fb{q}-\Fbp{q}}{F}$, and $\mathbf{e} \triangleq \begin{bmatrix} e_1 \hdots e_Q\end{bmatrix}^T$ with $e_q =\norm{\Eb{q}}{F}$, since $\sigmsq{\Hb{qq}^{-1}\Hb{qr}} = \sr{\Hb{qr}^H\Hb{qq}^{-H}\Hb{qq}^{-1}\Hb{qr}}$, 
    from the definition of $\Sc$ and $\mathbf{e}$,
    \begin{align}
        f_q \leq \br{\pa{\mathbf{I}_Q + \Sc}\mathbf{e}}_{q} \, \implies \, \norm{\mathbf{f}}{2}&\leq \norm{\pa{\mathbf{I}_Q + \Sc}\mathbf{e}}{2} \\&\leq \norm{\mathbf{I}_Q + \Sc}{2}\norm{\mathbf{e}}{2},
    \end{align}
    where the second inequality comes from the submultiplicative property of induced norms \cite[Corollary 9.4.4]{bernstein_dennis_s._matrix_2009}. As the norm of $\mathbf{f}$ and $\mathbf{e}$ are equal respectively to the left and right norm of \eqref{eq:Lipschitz}, from the link between 2-norms and singular values \cite[Prop. 9.4.9]{bernstein_dennis_s._matrix_2009}, the lemma is obtained.
\end{proof}
 \begin{lemma} If all the channel matrices $\Hb{qq}$ are square full rank, and if $\sr{\Sc^{s}} < 1$, the operator $\Faggb$ is strongly monotone with constant $\mu = 1 -\sr{\Sc^{s}}$: $\forall \Qaggb, \, \Qaggb^{'} \in \setAggb$,
        \begin{equation}
            \Tr{\pa{\Faggb - \Faggbp}^H\pa{\Qaggb - \Qaggb^{'}}}\geq \mu \norm{\Qaggb-\Qaggb^{'}}{F}^2.
            \label{eq:fullStrongMono}
        \end{equation}
        \label{lemma:monoton}
        \end{lemma}
\begin{proof}
The $q$th term of the LHS of \eqref{eq:fullStrongMono} can be bounded as
        \begin{align}
        &\Tr{\pa{\Fb{q}-\Fbp{q}}\Eb{q}}\quad \quad  \quad \quad \nonumber\\ 
             & \geq  \norm{\Eb{q}}{F}^2-\sum_{r\neq q}\abs{\Tr{\Hb{qq}^{-1}\Hb{qr}\Eb{r}\Hb{qr}^H\Hb{qq}^{-H}\Eb{q}}},\label{eq:triangineqtrace}
        \end{align}

where the lower bound comes from the triangular inequality. Moreover, from \cite[Fact 9.3.9, Corollary 9.3.7]{bernstein_dennis_s._matrix_2009}, \begin{align}
 &\abs{\Tr{\Hb{qq}^{-1}\Hb{qr}\Eb{r}\Hb{qr}^H\Hb{qq}^{-H}\Eb{q}}} \nonumber\\& \qquad \leq \norm{\Hb{qq}^{-1}\Hb{qr}\Eb{r}}{F}\norm{\Hb{qr}^H\Hb{qq}^{-H}\Eb{q}}{F},\label{eq:boundabstrace}\\
            & \qquad \leq \sigmsq{\Hb{qq}^{-1}\Hb{qr}}\norm{\Eb{r}}{F}\norm{\Eb{q}}{F}\label{eq:boundsingprod}.
        \end{align}
Grouping \eqref{eq:triangineqtrace} and \eqref{eq:boundsingprod} for all $q\in \Omega$,
\begin{align}
    &\Tr{\pa{\Faggb - \Faggbp}^H\pa{\Qaggb - \Qaggb^{'}}}\nonumber\\
     & \qquad \qquad \qquad \qquad \qquad \qquad\geq \mathbf{e}^T\pa{\mathbf{I}_{Q}-\Sc}\mathbf{e},\\
    & \qquad \qquad \qquad \qquad \qquad \qquad\geq \lambda_{\min}\pa{\mathbf{I}_{Q}-\Sc^s} \norm{\mathbf{e}}{2}^2
\end{align} 
with $\mathbf{e}$ being defined in the proof of \Cref{lemma:Lipschitz} and the second inequality following from \cite[Fact 3.7.5, Fact 8.15.17]{bernstein_dennis_s._matrix_2009}. Since $\lambda_{\min}\pa{\mathbf{I}_Q-\Sc^{s}} = 1-\lambda_{\max}\pa{\Sc^{s}}$ \cite[Prop. 4.4.5.ix]{bernstein_dennis_s._matrix_2009} and since the eigenvalues of $\Sc^s$ are real, the lemma is obtained. 
\end{proof}

\begin{lemma}
$\forall \Qaggb,\Qaggb^{'},\overline{\mathbf{Y}}\in \fullsetEEAggb$, the strategy dependent sets are such that 
\begin{align}
    \norm{\br{\overline{\mathbf{Y}}}_{\setEEAggbs}-\br{\overline{\mathbf{Y}}}_{\setEEAggbps}}{F} \leq  \norm{\PEEaggb - \PEEaggbp}{2}.
\end{align}
\label{lemma:smoothpower}
\end{lemma}
\begin{proof}
As $\setEEAggbs$ is a Cartesian product of $Q$ sets, the projection can be decomposed in $Q$ projections to give:
     \begin{align}
         &\norm{\br{\overline{\mathbf{Y}}}_{\setEEAggbs}-\br{\overline{\mathbf{Y}}}_{\setEEAggbps}}{F}^2 \nonumber \\&= \sum_{q=1}^Q \norm{\br{\overline{\mathbf{Y}}_q}_{\Qbareq{\PEEb{q}}}-\br{\overline{\mathbf{Y}}_q}_{\Qbareq{\PEEbp{q}}}}{F}^2.
         \label{eq:smoothfull}
     \end{align}
     From the equivalence between the projection \eqref{eq:EE_projection} and the waterfilling \eqref{eq:waterfilling_EE}, defining the EVD of $\overline{\mathbf{Y}}_q$ as $\mathbf{U}_{Y,q}\pa{-\mathbf{D}_{Y,q}}\mathbf{U}_{Y,q}^H$, the following holds:
     \begin{align}
         & \br{\overline{\mathbf{Y}}_q}_{\Qbareq{\PEEb{q}}}  = \mathbf{U}_{Y,q}\pa{\mu_q\mathbf{I}_{r_q}-\mathbf{D}_{Y,q}}^+\mathbf{U}_{Y,q}^H \nonumber\\&  \text{with}\, \mu_q\,  \st\quad  \Tr{\pa{\mu_q\mathbf{I}_{r_q}-\mathbf{D}_{Y,q}}^+}=\PEEb{q}.
     \end{align}
     It should be noted that the dependency on $\Qaggb$ and $\Qaggb^{'}$ is limited to $\mu_q$, as $\mathbf{U}_{Y,q}$ and $\mathbf{D}_{Y,q}$ only depend on $\overline{\mathbf{Y}}$.  Hence, letting $\mathbf{d}_{Y,q} = \text{diag}\pa{\mathbf{D}_{Y,q}}$, the $q$th term of \eqref{eq:smoothfull} can be bounded as follows.
     \begin{align}
         & \norm{\br{\overline{\mathbf{Y}}_q}_{\Qbareq{\PEEb{q}}}-\br{\overline{\mathbf{Y}}_q}_{\Qbareq{\PEEbp{q}}}}{F}\nonumber\\ &=\norm{\pa{\mu_q\mathbf{I}_{r_q}-\mathbf{D}_{Y,q}}^+-\pa{\mu_q^{'}\mathbf{I}_{r_q}-\mathbf{D}_{Y,q}}^+}{F}, \label{eq:smooth_frobinv}\\
         &  = \norm{\pa{\mu_q\mathbf{1}_{r_q\times 1}-\mathbf{d}_{Y,q}}^+-\pa{\mu_q^{'}\mathbf{1}_{r_q\times 1}-\mathbf{d}_{Y,q}}^+}{2},\label{eq:smooth_diag}\\
         &\leq \norm{\pa{\mu_q\mathbf{1}_{r_q\times 1}-\mathbf{d}_{Y,q}}^+-\pa{\mu_q^{'}\mathbf{1}_{r_q\times 1}-\mathbf{d}_{Y,q}}^+}{1},
         \label{eq:smooth_12norm}
     \end{align}
     where \eqref{eq:smooth_frobinv} follows from the unitary invariance of the Frobenius norm, \eqref{eq:smooth_diag} from the diagonal structure of the matrices and  \eqref{eq:smooth_12norm} from the inequality between 2-norms and 1-norms \cite[Prop. 9.1.5]{bernstein_dennis_s._matrix_2009}. Without loss of generality, suppose that $\mu_q \geq \mu_q^{'}$ (corresponding to $\PEEb{q}\geq \hat{P}_q(\Qb{-q}^{'})$). In that case, all the vector elements of \eqref{eq:smooth_12norm} are positive, the 1-norm hence boiling down to 
     \begin{align}
    &\sum_{n=1}^{r_q}\pa{ \pa{\mu_q-\br{\mathbf{d}_{Y,q}}_n}^+ - \pa{\mu_q^{'}-\br{\mathbf{d}_{Y,q}}_n}^+}\nonumber\\
         & \qquad \qquad \qquad\qquad \qquad=  \PEEb{q}-\PEEbp{q}.
     \end{align}
     Grouping together the $Q$ bounds, the lemma follows.
\end{proof}




\ifCLASSOPTIONcaptionsoff
  \newpage
\fi



\bibliographystyle{IEEEtran}
\bibliography{IEEEabrv,bare_jnrl.bib}

\begin{thebibliography}{10}
\providecommand{\url}[1]{#1}
\csname url@samestyle\endcsname
\providecommand{\newblock}{\relax}
\providecommand{\bibinfo}[2]{#2}
\providecommand{\BIBentrySTDinterwordspacing}{\spaceskip=0pt\relax}
\providecommand{\BIBentryALTinterwordstretchfactor}{4}
\providecommand{\BIBentryALTinterwordspacing}{\spaceskip=\fontdimen2\font plus
\BIBentryALTinterwordstretchfactor\fontdimen3\font minus
  \fontdimen4\font\relax}
\providecommand{\BIBforeignlanguage}[2]{{%
\expandafter\ifx\csname l@#1\endcsname\relax
\typeout{** WARNING: IEEEtran.bst: No hyphenation pattern has been}%
\typeout{** loaded for the language `#1'. Using the pattern for}%
\typeout{** the default language instead.}%
\else
\language=\csname l@#1\endcsname
\fi
#2}}
\providecommand{\BIBdecl}{\relax}
\BIBdecl

\bibitem{li2019computation}
Q.~Li, J.~Zhao, and Y.~Gong, ``Computation offloading and resource allocation
  for mobile edge computing with multiple access points,'' \emph{IET
  Communications}, vol.~13, no.~17, pp. 2668--2677, 2019.

\bibitem{liu2019joint}
J.~Liu, P.~Li, J.~Liu, and J.~Lai, ``Joint offloading and transmission power
  control for mobile edge computing,'' \emph{IEEE Access}, vol.~7, pp.
  81\,640--81\,651, 2019.

\bibitem{moura2018game}
J.~Moura and D.~Hutchison, ``Game theory for multi-access edge computing:
  Survey, use cases, and future trends,'' \emph{IEEE Communications Surveys \&
  Tutorials}, vol.~21, no.~1, pp. 260--288, 2018.

\bibitem{scutari_optimal_2007}
\BIBentryALTinterwordspacing
G.~Scutari, D.~P. Palomar, and S.~Barbarossa, ``Optimal {Linear} {Precoding}
  {Strategies} for {Wideband} {Non}-{Cooperative} {Systems} based on {Game}
  {Theory}-{Part} {I}: {Nash} {Equilibria},'' \emph{CoRR}, vol. abs/0707.0568,
  2007. [Online]. Available: \url{http://arxiv.org/abs/0707.0568}
\BIBentrySTDinterwordspacing

\bibitem{scutari_competitive_2008}
------, ``Competitive {Design} of {Multiuser} {MIMO} {Systems} {Based} on
  {Game} {Theory}: {A} {Unified} {View},'' \emph{IEEE Journal on Selected Areas
  in Communications}, vol.~26, no.~7, pp. 1089--1103, Sep. 2008.

\bibitem{ren2010distributed}
S.~Ren and M.~Van Der~Schaar, ``Distributed power allocation in multi-user
  multi-channel cellular relay networks,'' \emph{IEEE Transactions on Wireless
  Communications}, vol.~9, no.~6, pp. 1952--1964, 2010.

\bibitem{etkin2007spectrum}
R.~Etkin, A.~Parekh, and D.~Tse, ``Spectrum sharing for unlicensed bands,''
  \emph{IEEE Journal on selected areas in communications}, vol.~25, no.~3, pp.
  517--528, 2007.

\bibitem{scutari_optimal_2008}
G.~Scutari, D.~P. Palomar, and S.~Barbarossa, ``Optimal {Linear} {Precoding}
  {Strategies} for {Wideband} {Non}-{Cooperative} {Systems} {Based} on {Game}
  {Theory}—{Part} {II}: {Algorithms},'' \emph{IEEE Transactions on Signal
  Processing}, vol.~56, no.~3, pp. 1250--1267, Mar. 2008.

\bibitem{yu2002distributed}
W.~Yu, G.~Ginis, and J.~M. Cioffi, ``Distributed multiuser power control for
  digital subscriber lines,'' \emph{IEEE Journal on Selected areas in
  Communications}, vol.~20, no.~5, pp. 1105--1115, 2002.

\bibitem{cendrillon2007autonomous}
R.~Cendrillon, J.~Huang, M.~Chiang, and M.~Moonen, ``Autonomous spectrum
  balancing for digital subscriber lines,'' \emph{IEEE Transactions on Signal
  Processing}, vol.~55, no.~8, pp. 4241--4257, 2007.

\bibitem{shum2007convergence}
K.~W. Shum, K.-K. Leung, and C.~W. Sung, ``Convergence of iterative
  waterfilling algorithm for gaussian interference channels,'' \emph{IEEE
  Journal on Selected Areas in Communications}, vol.~25, no.~6, pp. 1091--1100,
  2007.

\bibitem{luo2006analysis}
Z.-Q. Luo and J.-S. Pang, ``Analysis of iterative waterfilling algorithm for
  multiuser power control in digital subscriber lines,'' \emph{EURASIP Journal
  on Advances in Signal Processing}, vol. 2006, no.~1, p. 024012, 2006.

\bibitem{yamashita2004nonlinear}
N.~Yamashita* and Z.-Q. Luo, ``A nonlinear complementarity approach to
  multiuser power control for digital subscriber lines,'' \emph{Optimization
  Methods and Software}, vol.~19, no.~5, pp. 633--652, 2004.

\bibitem{scutari2010convex}
G.~Scutari, D.~P. Palomar, F.~Facchinei, and J.-S. Pang, ``Convex optimization,
  game theory, and variational inequality theory,'' \emph{IEEE Signal
  Processing Magazine}, vol.~27, no.~3, pp. 35--49, 2010.

\bibitem{pang_design_2010}
J.-S. Pang, G.~Scutari, D.~Palomar, and F.~Facchinei, ``Design of {Cognitive}
  {Radio} {Systems} {Under} {Temperature}-{Interference} {Constraints}: {A}
  {Variational} {Inequality} {Approach},'' \emph{Signal Processing, IEEE
  Transactions on}, vol.~58, pp. 3251 -- 3271, 2010.

\bibitem{scutari_mimo_2009}
G.~Scutari, D.~P. Palomar, and S.~Barbarossa, ``The {MIMO} {Iterative}
  {Waterfilling} {Algorithm},'' \emph{IEEE Transactions on Signal Processing},
  vol.~57, no.~5, pp. 1917--1935, May 2009.

\bibitem{scutari2009mimo}
G.~Scutari and D.~P. Palomar, ``Mimo cognitive radio: A game theoretical
  approach,'' \emph{IEEE Transactions on Signal Processing}, vol.~58, no.~2,
  pp. 761--780, 2009.

\bibitem{BacciEECompetitive}
G.~{Bacci}, E.~V. {Belmega}, P.~{Mertikopoulos}, and L.~{Sanguinetti},
  ``Energy-aware competitive power allocation for heterogeneous networks under
  qos constraints,'' \emph{IEEE Transactions on Wireless Communications},
  vol.~14, no.~9, pp. 4728--4742, 2015.

\bibitem{stupia_power_2015}
I.~Stupia, L.~Sanguinetti, G.~Bacci, and L.~Vandendorpe, ``Power {Control} in
  {Networks} {With} {Heterogeneous} {Users}: {A} {Quasi}-{Variational}
  {Inequality} {Approach},'' \emph{Signal Processing, IEEE Transactions on},
  vol.~63, pp. 5691--5705, 2015.

\bibitem{ZapponeCompetitiveEE}
A.~{Zappone}, E.~A. {Jorswieck}, and S.~{Buzzi}, ``Energy efficiency and
  interference neutralization in two-hop mimo interference channels,''
  \emph{IEEE Transactions on Signal Processing}, vol.~62, no.~24, pp.
  6481--6495, 2014.

\bibitem{miao2011distributed}
G.~Miao, N.~Himayat, G.~Y. Li, and S.~Talwar, ``Distributed interference-aware
  energy-efficient power optimization,'' \emph{IEEE Transactions on Wireless
  Communications}, vol.~10, no.~4, pp. 1323--1333, 2011.

\bibitem{zhong2013energy}
W.~Zhong and J.~Wang, ``Energy efficient spectrum sharing strategy selection
  for cognitive mimo interference channels,'' \emph{IEEE Transactions on Signal
  Processing}, vol.~61, no.~14, pp. 3705--3717, 2013.

\bibitem{buzzi2011potential}
S.~Buzzi, G.~Colavolpe, D.~Saturnino, and A.~Zappone, ``Potential games for
  energy-efficient power control and subcarrier allocation in uplink multicell
  ofdma systems,'' \emph{IEEE Journal of Selected Topics in Signal Processing},
  vol.~6, no.~2, pp. 89--103, 2011.

\bibitem{pan2014totally}
C.~Pan, W.~Zhang, B.~Du, H.~Ren, and M.~Chen, ``Totally distributed
  energy-efficient transmission design in mimo interference channels,'' in
  \emph{2014 IEEE Global Communications Conference}.\hskip 1em plus 0.5em minus
  0.4em\relax IEEE, 2014, pp. 3964--3969.

\bibitem{goldsmith2003capacity}
A.~Goldsmith, S.~A. Jafar, N.~Jindal, and S.~Vishwanath, ``Capacity limits of
  mimo channels,'' \emph{IEEE Journal on selected areas in Communications},
  vol.~21, no.~5, pp. 684--702, 2003.

\bibitem{horn_matrix_1990}
R.~Horn and C.~Johnson, \emph{Matrix {Analysis}}.\hskip 1em plus 0.5em minus
  0.4em\relax Cambridge University Press, 1990.

\bibitem{xu_energy_2013}
J.~Xu and L.~Qiu, ``Energy {Efficiency} {Optimization} for {MIMO} {Broadcast}
  {Channels},'' \emph{IEEE Transactions on Wireless Communications}, vol.~12,
  no.~2, pp. 690--701, Feb. 2013.

\bibitem{dinkelbach_nonlinear_1967}
\BIBentryALTinterwordspacing
W.~Dinkelbach, ``On {Nonlinear} {Fractional} {Programming},'' \emph{Management
  Science}, vol.~13, no.~7, pp. 492--498, 1967. [Online]. Available:
  \url{https://doi.org/10.1287/mnsc.13.7.492}
\BIBentrySTDinterwordspacing

\bibitem{schaible_fractional_1983}
\BIBentryALTinterwordspacing
S.~Schaible, ``Fractional programming,'' \emph{Zeitschrift für
  Operations-Research}, vol.~27, no.~1, pp. 39--54, Dec. 1983. [Online].
  Available: \url{https://doi.org/10.1007/BF01916898}
\BIBentrySTDinterwordspacing

\bibitem{Zappone_EE}
A.~Zappone and E.~Jorswieck, ``Energy efficiency in wireless networks via
  fractional programming theory,'' \emph{Foundations and Trends® in
  Communications and Information Theory}, vol.~11, pp. 185--396, 01 2015.

\bibitem{facchinei2007generalized}
F.~Facchinei and C.~Kanzow, ``Generalized nash equilibrium problems,''
  \emph{4or}, vol.~5, no.~3, pp. 173--210, 2007.

\bibitem{existencepureNE}
D.~Fudenberg and J.~Tirole, ``Game theory,'' 1991.

\bibitem{facchinei_finite-dimensional_2004}
F.~Facchinei and J.-S. Pang, Eds.,
  \emph{\BIBforeignlanguage{en}{Finite-{Dimensional} {Variational}
  {Inequalities} and {Complementarity} {Problems}}}, ser. Springer {Series} in
  {Operations} {Research} and {Financial} {Engineering}.\hskip 1em plus 0.5em
  minus 0.4em\relax New York, NY: Springer New York, 2004.

\bibitem{boyd_convex_2004}
\BIBentryALTinterwordspacing
S.~Boyd, S.~Boyd, L.~Vandenberghe, and C.~U. Press, \emph{Convex
  {Optimization}}, ser. Berichte über verteilte messysteme.\hskip 1em plus
  0.5em minus 0.4em\relax Cambridge University Press, 2004. [Online].
  Available: \url{https://books.google.be/books?id=mYm0bLd3fcoC}
\BIBentrySTDinterwordspacing

\bibitem{nesterov_solving_2007}
Y.~NESTEROV and L.~Scrimali, ``Solving strongly monotone variational and
  quasi-variational inequalities,'' \emph{Discrete and Continuous Dynamical
  Systems}, vol.~31, 2007.

\bibitem{bernstein_dennis_s._matrix_2009}
{Bernstein Dennis S.}, \emph{\BIBforeignlanguage{ENGL}{Matrix {Mathematics},
  {Theory}, {Facts}, and {Formulas} - {Second} {Edition}}}.\hskip 1em plus
  0.5em minus 0.4em\relax Princeton: Princeton University Press, 2009.

\bibitem{bertsekas_parallel_1989}
D.~Bertsekas and J.~Tsitsiklis, ``Parallel and distributed computation :
  numerical methods / {Dimitri} {P}. {Bertsekas}, {John} {N}. {Tsitsiklis},''
  \emph{SERBIULA (sistema Librum 2.0)}, 1989.

\bibitem{schwenk_tight_1986}
\BIBentryALTinterwordspacing
A.~J. Schwenk, ``Tight bounds on the spectral radius of asymmetric nonnegative
  matrices,'' \emph{Linear Algebra and its Applications}, vol.~75, pp. 257 --
  265, 1986. [Online]. Available:
  \url{http://www.sciencedirect.com/science/article/pii/002437958690193X}
\BIBentrySTDinterwordspacing

\end{thebibliography}
%



%
\ifarchiv
\else
\begin{IEEEbiography}{Guillaume Thiran} 
\end{IEEEbiography}

\begin{IEEEbiography}{Ivan Stupia}
Biography text here.
\end{IEEEbiography}


\begin{IEEEbiography}{Luc Vandendorpe}
Biography text here.
\end{IEEEbiography}
\fi




\end{document}